\documentclass[runningheads]{llncs}
\usepackage[T1]{fontenc}
\usepackage{comment}

\usepackage{environ}
\usepackage{xspace}

\usepackage{algorithm}
\usepackage{algpseudocode}
\usepackage{amscd}
\usepackage{amsfonts}
\usepackage{amsmath}
\usepackage{amssymb}
\usepackage{booktabs}
\usepackage{color}
\usepackage{datetime}
\usepackage{epsfig}
\usepackage{graphicx}
\usepackage[pdfa,unicode]{hyperref}
\usepackage[capitalise]{cleveref}
\usepackage{latexsym}
\usepackage{marginnote}
\usepackage{multirow}
\usepackage{paralist}
\usepackage{pifont}
\usepackage{placeins}
\usepackage{rotate}
\usepackage[draft]{todonotes}
\usepackage{tikz}
\usepackage{url}
\usepackage{varwidth}
\usepackage{verbatim} 
\usepackage{wrapfig}
\usepackage{xcolor,colortbl}
\usepackage{xspace}
 \usetikzlibrary{calc}
\usepackage{setspace}
\usepackage{subcaption}
\usepackage[inline]{enumitem}
\usepackage{soul}
\usepackage{microtype}
\usepackage[subtle]{savetrees}

\Crefname{definition}{Def.}{Defs.}

\makeatletter%
\@mparswitchfalse%
\makeatother%
\reversemarginpar%

\definecolor {snow}                {rgb}{1.00,0.98,0.98}
\definecolor {ghostwhite}          {rgb}{0.97,0.97,1.00}
\definecolor {whitesmoke}          {rgb}{0.96,0.96,0.96}
\definecolor {gainsboro}           {rgb}{0.86,0.86,0.86}
\definecolor {floralwhite}         {rgb}{1.00,0.98,0.94}
\definecolor {oldlace}             {rgb}{0.99,0.96,0.90}
\definecolor {linen}               {rgb}{0.98,0.94,0.90}
\definecolor {antiquewhite}        {rgb}{0.98,0.92,0.84}
\definecolor {papayawhip}          {rgb}{1.00,0.94,0.84}
\definecolor {blanchedalmond}      {rgb}{1.00,0.92,0.80}
\definecolor {bisque}              {rgb}{1.00,0.89,0.77}
\definecolor {peachpuff}           {rgb}{1.00,0.85,0.73}
\definecolor {navajowhite}         {rgb}{1.00,0.87,0.68}
\definecolor {moccasin}            {rgb}{1.00,0.89,0.71}
\definecolor {cornsilk}            {rgb}{1.00,0.97,0.86}
\definecolor {ivory}               {rgb}{1.00,1.00,0.94}
\definecolor {lemonchiffon}        {rgb}{1.00,0.98,0.80}
\definecolor {seashell}            {rgb}{1.00,0.96,0.93}
\definecolor {honeydew}            {rgb}{0.94,1.00,0.94}
\definecolor {mintcream}           {rgb}{0.96,1.00,0.98}
\definecolor {azure}               {rgb}{0.94,1.00,1.00}
\definecolor {aliceblue}           {rgb}{0.94,0.97,1.00}
\definecolor {lavender}            {rgb}{0.90,0.90,0.98}
\definecolor {lavenderblush}       {rgb}{1.00,0.94,0.96}
\definecolor {mistyrose}           {rgb}{1.00,0.89,0.88}
\definecolor {white}               {rgb}{1.00,1.00,1.00}
\definecolor {black}               {rgb}{0.00,0.00,0.00}
\definecolor {darkslategray}       {rgb}{0.18,0.31,0.31}
\definecolor {dimgray}             {rgb}{0.41,0.41,0.41}
\definecolor {slategray}           {rgb}{0.44,0.50,0.56}
\definecolor {lightslategray}      {rgb}{0.47,0.53,0.60}
\definecolor {gray}                {rgb}{0.75,0.75,0.75}
\definecolor {lightgrey}           {rgb}{0.83,0.83,0.83}
\definecolor {midnightblue}        {rgb}{0.10,0.10,0.44}
\definecolor {navy}                {rgb}{0.00,0.00,0.50}
\definecolor {cornflowerblue}      {rgb}{0.39,0.58,0.93}
\definecolor {darkslateblue}       {rgb}{0.28,0.24,0.55}
\definecolor {slateblue}           {rgb}{0.42,0.35,0.80}
\definecolor {mediumslateblue}     {rgb}{0.48,0.41,0.93}
\definecolor {lightslateblue}      {rgb}{0.52,0.44,1.00}
\definecolor {mediumblue}          {rgb}{0.00,0.00,0.80}
\definecolor {royalblue}           {rgb}{0.25,0.41,0.88}
\definecolor {blue}                {rgb}{0.00,0.00,1.00}
\definecolor {dodgerblue}          {rgb}{0.12,0.56,1.00}
\definecolor {deepskyblue}         {rgb}{0.00,0.75,1.00}
\definecolor {skyblue}             {rgb}{0.53,0.81,0.92}
\definecolor {lightskyblue}        {rgb}{0.53,0.81,0.98}
\definecolor {steelblue}           {rgb}{0.27,0.51,0.71}
\definecolor {lightsteelblue}      {rgb}{0.69,0.77,0.87}
\definecolor {lightblue}           {rgb}{0.68,0.85,0.90}
\definecolor {powderblue}          {rgb}{0.69,0.88,0.90}
\definecolor {paleturquoise}       {rgb}{0.69,0.93,0.93}
\definecolor {darkturquoise}       {rgb}{0.00,0.81,0.82}
\definecolor {mediumturquoise}     {rgb}{0.28,0.82,0.80}
\definecolor {turquoise}           {rgb}{0.25,0.88,0.82}
\definecolor {cyan}                {rgb}{0.00,1.00,1.00}
\definecolor {lightcyan}           {rgb}{0.88,1.00,1.00}
\definecolor {cadetblue}           {rgb}{0.37,0.62,0.63}
\definecolor {mediumaquamarine}    {rgb}{0.40,0.80,0.67}
\definecolor {aquamarine}          {rgb}{0.50,1.00,0.83}
\definecolor {darkgreen}           {rgb}{0.00,0.39,0.00}
\definecolor {darkolivegreen}      {rgb}{0.33,0.42,0.18}
\definecolor {darkseagreen}        {rgb}{0.56,0.74,0.56}
\definecolor {seagreen}            {rgb}{0.18,0.55,0.34}
\definecolor {mediumseagreen}      {rgb}{0.24,0.70,0.44}
\definecolor {lightseagreen}       {rgb}{0.13,0.70,0.67}
\definecolor {palegreen}           {rgb}{0.60,0.98,0.60}
\definecolor {springgreen}         {rgb}{0.00,1.00,0.50}
\definecolor {lawngreen}           {rgb}{0.49,0.99,0.00}
\definecolor {green}               {rgb}{0.00,1.00,0.00}
\definecolor {chartreuse}          {rgb}{0.50,1.00,0.00}
\definecolor {mediumspringgreen}   {rgb}{0.00,0.98,0.60}
\definecolor {greenyellow}         {rgb}{0.68,1.00,0.18}
\definecolor {limegreen}           {rgb}{0.20,0.80,0.20}
\definecolor {yellowgreen}         {rgb}{0.60,0.80,0.20}
\definecolor {forestgreen}         {rgb}{0.13,0.55,0.13}
\definecolor {olivedrab}           {rgb}{0.42,0.56,0.14}
\definecolor {darkkhaki}           {rgb}{0.74,0.72,0.42}
\definecolor {khaki}               {rgb}{0.94,0.90,0.55}
\definecolor {palegoldenrod}       {rgb}{0.93,0.91,0.67}
\definecolor {lightgoldenrodyellow} {rgb}{0.98,0.98,0.82}
\definecolor {lightyellow}         {rgb}{1.00,1.00,0.88}
\definecolor {yellow}              {rgb}{1.00,1.00,0.00}
\definecolor {gold}                {rgb}{1.00,0.84,0.00}
\definecolor {lightgoldenrod}      {rgb}{0.93,0.87,0.51}
\definecolor {goldenrod}           {rgb}{0.85,0.65,0.13}
\definecolor {darkgoldenrod}       {rgb}{0.72,0.53,0.04}
\definecolor {rosybrown}           {rgb}{0.74,0.56,0.56}
\definecolor {indianred}           {rgb}{0.80,0.36,0.36}
\definecolor {saddlebrown}         {rgb}{0.55,0.27,0.07}
\definecolor {sienna}              {rgb}{0.63,0.32,0.18}
\definecolor {peru}                {rgb}{0.80,0.52,0.25}
\definecolor {burlywood}           {rgb}{0.87,0.72,0.53}
\definecolor {beige}               {rgb}{0.96,0.96,0.86}
\definecolor {wheat}               {rgb}{0.96,0.87,0.70}
\definecolor {sandybrown}          {rgb}{0.96,0.64,0.38}
\definecolor {tan}                 {rgb}{0.82,0.71,0.55}
\definecolor {chocolate}           {rgb}{0.82,0.41,0.12}
\definecolor {firebrick}           {rgb}{0.70,0.13,0.13}
\definecolor {brown}               {rgb}{0.65,0.16,0.16}
\definecolor {darksalmon}          {rgb}{0.91,0.59,0.48}
\definecolor {salmon}              {rgb}{0.98,0.50,0.45}
\definecolor {lightsalmon}         {rgb}{1.00,0.63,0.48}
\definecolor {orange}              {rgb}{1.00,0.65,0.00}
\definecolor {darkorange}          {rgb}{1.00,0.55,0.00}
\definecolor {coral}               {rgb}{1.00,0.50,0.31}
\definecolor {lightcoral}          {rgb}{0.94,0.50,0.50}
\definecolor {tomato}              {rgb}{1.00,0.39,0.28}
\definecolor {orangered}           {rgb}{1.00,0.27,0.00}
\definecolor {red}                 {rgb}{1.00,0.00,0.00}
\definecolor {hotpink}             {rgb}{1.00,0.41,0.71}
\definecolor {deeppink}            {rgb}{1.00,0.08,0.58}
\definecolor {pink}                {rgb}{1.00,0.75,0.80}
\definecolor {lightpink}           {rgb}{1.00,0.71,0.76}
\definecolor {palevioletred}       {rgb}{0.86,0.44,0.58}
\definecolor {maroon}              {rgb}{0.69,0.19,0.38}
\definecolor {mediumvioletred}     {rgb}{0.78,0.08,0.52}
\definecolor {violetred}           {rgb}{0.82,0.13,0.56}
\definecolor {magenta}             {rgb}{1.00,0.00,1.00}
\definecolor {violet}              {rgb}{0.93,0.51,0.93}
\definecolor {plum}                {rgb}{0.87,0.63,0.87}
\definecolor {orchid}              {rgb}{0.85,0.44,0.84}
\definecolor {mediumorchid}        {rgb}{0.73,0.33,0.83}
\definecolor {darkorchid}          {rgb}{0.60,0.20,0.80}
\definecolor {darkviolet}          {rgb}{0.58,0.00,0.83}
\definecolor {blueviolet}          {rgb}{0.54,0.17,0.89}
\definecolor {purple}              {rgb}{0.63,0.13,0.94}
\definecolor {mediumpurple}        {rgb}{0.58,0.44,0.86}
\definecolor {thistle}             {rgb}{0.85,0.75,0.85}
\definecolor {snow2}               {rgb}{0.93,0.91,0.91}
\definecolor {snow3}               {rgb}{0.80,0.79,0.79}
\definecolor {snow4}               {rgb}{0.55,0.54,0.54}
\definecolor {seashell2}           {rgb}{0.93,0.90,0.87}
\definecolor {seashell3}           {rgb}{0.80,0.77,0.75}
\definecolor {seashell4}           {rgb}{0.55,0.53,0.51}
\definecolor {antiquewhite1}       {rgb}{1.00,0.94,0.86}
\definecolor {antiquewhite2}       {rgb}{0.93,0.87,0.80}
\definecolor {antiquewhite3}       {rgb}{0.80,0.75,0.69}
\definecolor {antiquewhite4}       {rgb}{0.55,0.51,0.47}
\definecolor {bisque2}             {rgb}{0.93,0.84,0.72}
\definecolor {bisque3}             {rgb}{0.80,0.72,0.62}
\definecolor {bisque4}             {rgb}{0.55,0.49,0.42}
\definecolor {peachpuff2}          {rgb}{0.93,0.80,0.68}
\definecolor {peachpuff3}          {rgb}{0.80,0.69,0.58}
\definecolor {peachpuff4}          {rgb}{0.55,0.47,0.40}
\definecolor {navajowhite2}        {rgb}{0.93,0.81,0.63}
\definecolor {navajowhite3}        {rgb}{0.80,0.70,0.55}
\definecolor {navajowhite4}        {rgb}{0.55,0.47,0.37}
\definecolor {lemonchiffon2}       {rgb}{0.93,0.91,0.75}
\definecolor {lemonchiffon3}       {rgb}{0.80,0.79,0.65}
\definecolor {lemonchiffon4}       {rgb}{0.55,0.54,0.44}
\definecolor {cornsilk2}           {rgb}{0.93,0.91,0.80}
\definecolor {cornsilk3}           {rgb}{0.80,0.78,0.69}
\definecolor {cornsilk4}           {rgb}{0.55,0.53,0.47}
\definecolor {ivory2}              {rgb}{0.93,0.93,0.88}
\definecolor {ivory3}              {rgb}{0.80,0.80,0.76}
\definecolor {ivory4}              {rgb}{0.55,0.55,0.51}
\definecolor {honeydew2}           {rgb}{0.88,0.93,0.88}
\definecolor {honeydew3}           {rgb}{0.76,0.80,0.76}
\definecolor {honeydew4}           {rgb}{0.51,0.55,0.51}
\definecolor {lavenderblush2}      {rgb}{0.93,0.88,0.90}
\definecolor {lavenderblush3}      {rgb}{0.80,0.76,0.77}
\definecolor {lavenderblush4}      {rgb}{0.55,0.51,0.53}
\definecolor {mistyrose2}          {rgb}{0.93,0.84,0.82}
\definecolor {mistyrose3}          {rgb}{0.80,0.72,0.71}
\definecolor {mistyrose4}          {rgb}{0.55,0.49,0.48}
\definecolor {azure2}              {rgb}{0.88,0.93,0.93}
\definecolor {azure3}              {rgb}{0.76,0.80,0.80}
\definecolor {azure4}              {rgb}{0.51,0.55,0.55}
\definecolor {slateblue1}          {rgb}{0.51,0.44,1.00}
\definecolor {slateblue2}          {rgb}{0.48,0.40,0.93}
\definecolor {slateblue3}          {rgb}{0.41,0.35,0.80}
\definecolor {slateblue4}          {rgb}{0.28,0.24,0.55}
\definecolor {royalblue1}          {rgb}{0.28,0.46,1.00}
\definecolor {royalblue2}          {rgb}{0.26,0.43,0.93}
\definecolor {royalblue3}          {rgb}{0.23,0.37,0.80}
\definecolor {royalblue4}          {rgb}{0.15,0.25,0.55}
\definecolor {blue2}               {rgb}{0.00,0.00,0.93}
\definecolor {blue4}               {rgb}{0.00,0.00,0.55}
\definecolor {dodgerblue2}         {rgb}{0.11,0.53,0.93}
\definecolor {dodgerblue3}         {rgb}{0.09,0.45,0.80}
\definecolor {dodgerblue4}         {rgb}{0.06,0.31,0.55}
\definecolor {steelblue1}          {rgb}{0.39,0.72,1.00}
\definecolor {steelblue2}          {rgb}{0.36,0.67,0.93}
\definecolor {steelblue3}          {rgb}{0.31,0.58,0.80}
\definecolor {steelblue4}          {rgb}{0.21,0.39,0.55}
\definecolor {deepskyblue2}        {rgb}{0.00,0.70,0.93}
\definecolor {deepskyblue3}        {rgb}{0.00,0.60,0.80}
\definecolor {deepskyblue4}        {rgb}{0.00,0.41,0.55}
\definecolor {skyblue1}            {rgb}{0.53,0.81,1.00}
\definecolor {skyblue2}            {rgb}{0.49,0.75,0.93}
\definecolor {skyblue3}            {rgb}{0.42,0.65,0.80}
\definecolor {skyblue4}            {rgb}{0.29,0.44,0.55}
\definecolor {lightskyblue1}       {rgb}{0.69,0.89,1.00}
\definecolor {lightskyblue2}       {rgb}{0.64,0.83,0.93}
\definecolor {lightskyblue3}       {rgb}{0.55,0.71,0.80}
\definecolor {lightskyblue4}       {rgb}{0.38,0.48,0.55}
\definecolor {slategray1}          {rgb}{0.78,0.89,1.00}
\definecolor {slategray2}          {rgb}{0.73,0.83,0.93}
\definecolor {slategray3}          {rgb}{0.62,0.71,0.80}
\definecolor {slategray4}          {rgb}{0.42,0.48,0.55}
\definecolor {lightsteelblue1}     {rgb}{0.79,0.88,1.00}
\definecolor {lightsteelblue2}     {rgb}{0.74,0.82,0.93}
\definecolor {lightsteelblue3}     {rgb}{0.64,0.71,0.80}
\definecolor {lightsteelblue4}     {rgb}{0.43,0.48,0.55}
\definecolor {lightblue1}          {rgb}{0.75,0.94,1.00}
\definecolor {lightblue2}          {rgb}{0.70,0.87,0.93}
\definecolor {lightblue3}          {rgb}{0.60,0.75,0.80}
\definecolor {lightblue4}          {rgb}{0.41,0.51,0.55}
\definecolor {lightcyan2}          {rgb}{0.82,0.93,0.93}
\definecolor {lightcyan3}          {rgb}{0.71,0.80,0.80}
\definecolor {lightcyan4}          {rgb}{0.48,0.55,0.55}
\definecolor {paleturquoise1}      {rgb}{0.73,1.00,1.00}
\definecolor {paleturquoise2}      {rgb}{0.68,0.93,0.93}
\definecolor {paleturquoise3}      {rgb}{0.59,0.80,0.80}
\definecolor {paleturquoise4}      {rgb}{0.40,0.55,0.55}
\definecolor {cadetblue1}          {rgb}{0.60,0.96,1.00}
\definecolor {cadetblue2}          {rgb}{0.56,0.90,0.93}
\definecolor {cadetblue3}          {rgb}{0.48,0.77,0.80}
\definecolor {cadetblue4}          {rgb}{0.33,0.53,0.55}
\definecolor {turquoise1}          {rgb}{0.00,0.96,1.00}
\definecolor {turquoise2}          {rgb}{0.00,0.90,0.93}
\definecolor {turquoise3}          {rgb}{0.00,0.77,0.80}
\definecolor {turquoise4}          {rgb}{0.00,0.53,0.55}
\definecolor {cyan2}               {rgb}{0.00,0.93,0.93}
\definecolor {cyan3}               {rgb}{0.00,0.80,0.80}
\definecolor {cyan4}               {rgb}{0.00,0.55,0.55}
\definecolor {darkslategray1}      {rgb}{0.59,1.00,1.00}
\definecolor {darkslategray2}      {rgb}{0.55,0.93,0.93}
\definecolor {darkslategray3}      {rgb}{0.47,0.80,0.80}
\definecolor {darkslategray4}      {rgb}{0.32,0.55,0.55}
\definecolor {aquamarine2}         {rgb}{0.46,0.93,0.78}
\definecolor {aquamarine4}         {rgb}{0.27,0.55,0.45}
\definecolor {darkseagreen1}       {rgb}{0.76,1.00,0.76}
\definecolor {darkseagreen2}       {rgb}{0.71,0.93,0.71}
\definecolor {darkseagreen3}       {rgb}{0.61,0.80,0.61}
\definecolor {darkseagreen4}       {rgb}{0.41,0.55,0.41}
\definecolor {seagreen1}           {rgb}{0.33,1.00,0.62}
\definecolor {seagreen2}           {rgb}{0.31,0.93,0.58}
\definecolor {seagreen3}           {rgb}{0.26,0.80,0.50}
\definecolor {palegreen1}          {rgb}{0.60,1.00,0.60}
\definecolor {palegreen2}          {rgb}{0.56,0.93,0.56}
\definecolor {palegreen3}          {rgb}{0.49,0.80,0.49}
\definecolor {palegreen4}          {rgb}{0.33,0.55,0.33}
\definecolor {springgreen2}        {rgb}{0.00,0.93,0.46}
\definecolor {springgreen3}        {rgb}{0.00,0.80,0.40}
\definecolor {springgreen4}        {rgb}{0.00,0.55,0.27}
\definecolor {green2}              {rgb}{0.00,0.93,0.00}
\definecolor {green3}              {rgb}{0.00,0.80,0.00}
\definecolor {green4}              {rgb}{0.00,0.55,0.00}
\definecolor {chartreuse2}         {rgb}{0.46,0.93,0.00}
\definecolor {chartreuse3}         {rgb}{0.40,0.80,0.00}
\definecolor {chartreuse4}         {rgb}{0.27,0.55,0.00}
\definecolor {olivedrab1}          {rgb}{0.75,1.00,0.24}
\definecolor {olivedrab2}          {rgb}{0.70,0.93,0.23}
\definecolor {olivedrab4}          {rgb}{0.41,0.55,0.13}
\definecolor {darkolivegreen1}     {rgb}{0.79,1.00,0.44}
\definecolor {darkolivegreen2}     {rgb}{0.74,0.93,0.41}
\definecolor {darkolivegreen3}     {rgb}{0.64,0.80,0.35}
\definecolor {darkolivegreen4}     {rgb}{0.43,0.55,0.24}
\definecolor {khaki1}              {rgb}{1.00,0.96,0.56}
\definecolor {khaki2}              {rgb}{0.93,0.90,0.52}
\definecolor {khaki3}              {rgb}{0.80,0.78,0.45}
\definecolor {khaki4}              {rgb}{0.55,0.53,0.31}
\definecolor {lightgoldenrod1}     {rgb}{1.00,0.93,0.55}
\definecolor {lightgoldenrod2}     {rgb}{0.93,0.86,0.51}
\definecolor {lightgoldenrod3}     {rgb}{0.80,0.75,0.44}
\definecolor {lightgoldenrod4}     {rgb}{0.55,0.51,0.30}
\definecolor {lightyellow2}        {rgb}{0.93,0.93,0.82}
\definecolor {lightyellow3}        {rgb}{0.80,0.80,0.71}
\definecolor {lightyellow4}        {rgb}{0.55,0.55,0.48}
\definecolor {yellow2}             {rgb}{0.93,0.93,0.00}
\definecolor {yellow3}             {rgb}{0.80,0.80,0.00}
\definecolor {yellow4}             {rgb}{0.55,0.55,0.00}
\definecolor {gold2}               {rgb}{0.93,0.79,0.00}
\definecolor {gold3}               {rgb}{0.80,0.68,0.00}
\definecolor {gold4}               {rgb}{0.55,0.46,0.00}
\definecolor {goldenrod1}          {rgb}{1.00,0.76,0.15}
\definecolor {goldenrod2}          {rgb}{0.93,0.71,0.13}
\definecolor {goldenrod3}          {rgb}{0.80,0.61,0.11}
\definecolor {goldenrod4}          {rgb}{0.55,0.41,0.08}
\definecolor {darkgoldenrod1}      {rgb}{1.00,0.73,0.06}
\definecolor {darkgoldenrod2}      {rgb}{0.93,0.68,0.05}
\definecolor {darkgoldenrod3}      {rgb}{0.80,0.58,0.05}
\definecolor {darkgoldenrod4}      {rgb}{0.55,0.40,0.03}
\definecolor {rosybrown1}          {rgb}{1.00,0.76,0.76}
\definecolor {rosybrown2}          {rgb}{0.93,0.71,0.71}
\definecolor {rosybrown3}          {rgb}{0.80,0.61,0.61}
\definecolor {rosybrown4}          {rgb}{0.55,0.41,0.41}
\definecolor {indianred1}          {rgb}{1.00,0.42,0.42}
\definecolor {indianred2}          {rgb}{0.93,0.39,0.39}
\definecolor {indianred3}          {rgb}{0.80,0.33,0.33}
\definecolor {indianred4}          {rgb}{0.55,0.23,0.23}
\definecolor {sienna1}             {rgb}{1.00,0.51,0.28}
\definecolor {sienna2}             {rgb}{0.93,0.47,0.26}
\definecolor {sienna3}             {rgb}{0.80,0.41,0.22}
\definecolor {sienna4}             {rgb}{0.55,0.28,0.15}
\definecolor {burlywood1}          {rgb}{1.00,0.83,0.61}
\definecolor {burlywood2}          {rgb}{0.93,0.77,0.57}
\definecolor {burlywood3}          {rgb}{0.80,0.67,0.49}
\definecolor {burlywood4}          {rgb}{0.55,0.45,0.33}
\definecolor {wheat1}              {rgb}{1.00,0.91,0.73}
\definecolor {wheat2}              {rgb}{0.93,0.85,0.68}
\definecolor {wheat3}              {rgb}{0.80,0.73,0.59}
\definecolor {wheat4}              {rgb}{0.55,0.49,0.40}
\definecolor {tan1}                {rgb}{1.00,0.65,0.31}
\definecolor {tan2}                {rgb}{0.93,0.60,0.29}
\definecolor {tan4}                {rgb}{0.55,0.35,0.17}
\definecolor {chocolate1}          {rgb}{1.00,0.50,0.14}
\definecolor {chocolate2}          {rgb}{0.93,0.46,0.13}
\definecolor {chocolate3}          {rgb}{0.80,0.40,0.11}
\definecolor {firebrick1}          {rgb}{1.00,0.19,0.19}
\definecolor {firebrick2}          {rgb}{0.93,0.17,0.17}
\definecolor {firebrick3}          {rgb}{0.80,0.15,0.15}
\definecolor {firebrick4}          {rgb}{0.55,0.10,0.10}
\definecolor {brown1}              {rgb}{1.00,0.25,0.25}
\definecolor {brown2}              {rgb}{0.93,0.23,0.23}
\definecolor {brown3}              {rgb}{0.80,0.20,0.20}
\definecolor {brown4}              {rgb}{0.55,0.14,0.14}
\definecolor {salmon1}             {rgb}{1.00,0.55,0.41}
\definecolor {salmon2}             {rgb}{0.93,0.51,0.38}
\definecolor {salmon3}             {rgb}{0.80,0.44,0.33}
\definecolor {salmon4}             {rgb}{0.55,0.30,0.22}
\definecolor {lightsalmon2}        {rgb}{0.93,0.58,0.45}
\definecolor {lightsalmon3}        {rgb}{0.80,0.51,0.38}
\definecolor {lightsalmon4}        {rgb}{0.55,0.34,0.26}
\definecolor {orange2}             {rgb}{0.93,0.60,0.00}
\definecolor {orange3}             {rgb}{0.80,0.52,0.00}
\definecolor {orange4}             {rgb}{0.55,0.35,0.00}
\definecolor {darkorange1}         {rgb}{1.00,0.50,0.00}
\definecolor {darkorange2}         {rgb}{0.93,0.46,0.00}
\definecolor {darkorange3}         {rgb}{0.80,0.40,0.00}
\definecolor {darkorange4}         {rgb}{0.55,0.27,0.00}
\definecolor {coral1}              {rgb}{1.00,0.45,0.34}
\definecolor {coral2}              {rgb}{0.93,0.42,0.31}
\definecolor {coral3}              {rgb}{0.80,0.36,0.27}
\definecolor {coral4}              {rgb}{0.55,0.24,0.18}
\definecolor {tomato2}             {rgb}{0.93,0.36,0.26}
\definecolor {tomato3}             {rgb}{0.80,0.31,0.22}
\definecolor {tomato4}             {rgb}{0.55,0.21,0.15}
\definecolor {orangered2}          {rgb}{0.93,0.25,0.00}
\definecolor {orangered3}          {rgb}{0.80,0.22,0.00}
\definecolor {orangered4}          {rgb}{0.55,0.15,0.00}
\definecolor {red2}                {rgb}{0.93,0.00,0.00}
\definecolor {red3}                {rgb}{0.80,0.00,0.00}
\definecolor {red4}                {rgb}{0.55,0.00,0.00}
\definecolor {deeppink2}           {rgb}{0.93,0.07,0.54}
\definecolor {deeppink3}           {rgb}{0.80,0.06,0.46}
\definecolor {deeppink4}           {rgb}{0.55,0.04,0.31}
\definecolor {hotpink1}            {rgb}{1.00,0.43,0.71}
\definecolor {hotpink2}            {rgb}{0.93,0.42,0.65}
\definecolor {hotpink3}            {rgb}{0.80,0.38,0.56}
\definecolor {hotpink4}            {rgb}{0.55,0.23,0.38}
\definecolor {pink1}               {rgb}{1.00,0.71,0.77}
\definecolor {pink2}               {rgb}{0.93,0.66,0.72}
\definecolor {pink3}               {rgb}{0.80,0.57,0.62}
\definecolor {pink4}               {rgb}{0.55,0.39,0.42}
\definecolor {lightpink1}          {rgb}{1.00,0.68,0.73}
\definecolor {lightpink2}          {rgb}{0.93,0.64,0.68}
\definecolor {lightpink3}          {rgb}{0.80,0.55,0.58}
\definecolor {lightpink4}          {rgb}{0.55,0.37,0.40}
\definecolor {palevioletred1}      {rgb}{1.00,0.51,0.67}
\definecolor {palevioletred2}      {rgb}{0.93,0.47,0.62}
\definecolor {palevioletred3}      {rgb}{0.80,0.41,0.54}
\definecolor {palevioletred4}      {rgb}{0.55,0.28,0.36}
\definecolor {maroon1}             {rgb}{1.00,0.20,0.70}
\definecolor {maroon2}             {rgb}{0.93,0.19,0.65}
\definecolor {maroon3}             {rgb}{0.80,0.16,0.56}
\definecolor {maroon4}             {rgb}{0.55,0.11,0.38}
\definecolor {violetred1}          {rgb}{1.00,0.24,0.59}
\definecolor {violetred2}          {rgb}{0.93,0.23,0.55}
\definecolor {violetred3}          {rgb}{0.80,0.20,0.47}
\definecolor {violetred4}          {rgb}{0.55,0.13,0.32}
\definecolor {magenta2}            {rgb}{0.93,0.00,0.93}
\definecolor {magenta3}            {rgb}{0.80,0.00,0.80}
\definecolor {magenta4}            {rgb}{0.55,0.00,0.55}
\definecolor {orchid1}             {rgb}{1.00,0.51,0.98}
\definecolor {orchid2}             {rgb}{0.93,0.48,0.91}
\definecolor {orchid3}             {rgb}{0.80,0.41,0.79}
\definecolor {orchid4}             {rgb}{0.55,0.28,0.54}
\definecolor {plum1}               {rgb}{1.00,0.73,1.00}
\definecolor {plum2}               {rgb}{0.93,0.68,0.93}
\definecolor {plum3}               {rgb}{0.80,0.59,0.80}
\definecolor {plum4}               {rgb}{0.55,0.40,0.55}
\definecolor {mediumorchid1}       {rgb}{0.88,0.40,1.00}
\definecolor {mediumorchid2}       {rgb}{0.82,0.37,0.93}
\definecolor {mediumorchid3}       {rgb}{0.71,0.32,0.80}
\definecolor {mediumorchid4}       {rgb}{0.48,0.22,0.55}
\definecolor {darkorchid1}         {rgb}{0.75,0.24,1.00}
\definecolor {darkorchid2}         {rgb}{0.70,0.23,0.93}
\definecolor {darkorchid3}         {rgb}{0.60,0.20,0.80}
\definecolor {darkorchid4}         {rgb}{0.41,0.13,0.55}
\definecolor {purple1}             {rgb}{0.61,0.19,1.00}
\definecolor {purple2}             {rgb}{0.57,0.17,0.93}
\definecolor {purple3}             {rgb}{0.49,0.15,0.80}
\definecolor {purple4}             {rgb}{0.33,0.10,0.55}
\definecolor {mediumpurple1}       {rgb}{0.67,0.51,1.00}
\definecolor {mediumpurple2}       {rgb}{0.62,0.47,0.93}
\definecolor {mediumpurple3}       {rgb}{0.54,0.41,0.80}
\definecolor {mediumpurple4}       {rgb}{0.36,0.28,0.55}
\definecolor {thistle1}            {rgb}{1.00,0.88,1.00}
\definecolor {thistle2}            {rgb}{0.93,0.82,0.93}
\definecolor {thistle3}            {rgb}{0.80,0.71,0.80}
\definecolor {thistle4}            {rgb}{0.55,0.48,0.55}
\definecolor {gray1}               {rgb}{0.01,0.01,0.01}
\definecolor {gray2}               {rgb}{0.02,0.02,0.02}
\definecolor {gray3}               {rgb}{0.03,0.03,0.03}
\definecolor {gray4}               {rgb}{0.04,0.04,0.04}
\definecolor {gray5}               {rgb}{0.05,0.05,0.05}
\definecolor {gray6}               {rgb}{0.06,0.06,0.06}
\definecolor {gray7}               {rgb}{0.07,0.07,0.07}
\definecolor {gray8}               {rgb}{0.08,0.08,0.08}
\definecolor {gray9}               {rgb}{0.09,0.09,0.09}
\definecolor {gray10}              {rgb}{0.10,0.10,0.10}
\definecolor {gray11}              {rgb}{0.11,0.11,0.11}
\definecolor {gray12}              {rgb}{0.12,0.12,0.12}
\definecolor {gray13}              {rgb}{0.13,0.13,0.13}
\definecolor {gray14}              {rgb}{0.14,0.14,0.14}
\definecolor {gray15}              {rgb}{0.15,0.15,0.15}
\definecolor {gray16}              {rgb}{0.16,0.16,0.16}
\definecolor {gray17}              {rgb}{0.17,0.17,0.17}
\definecolor {gray18}              {rgb}{0.18,0.18,0.18}
\definecolor {gray19}              {rgb}{0.19,0.19,0.19}
\definecolor {gray20}              {rgb}{0.20,0.20,0.20}
\definecolor {gray21}              {rgb}{0.21,0.21,0.21}
\definecolor {gray22}              {rgb}{0.22,0.22,0.22}
\definecolor {gray23}              {rgb}{0.23,0.23,0.23}
\definecolor {gray24}              {rgb}{0.24,0.24,0.24}
\definecolor {gray25}              {rgb}{0.25,0.25,0.25}
\definecolor {gray26}              {rgb}{0.26,0.26,0.26}
\definecolor {gray27}              {rgb}{0.27,0.27,0.27}
\definecolor {gray28}              {rgb}{0.28,0.28,0.28}
\definecolor {gray29}              {rgb}{0.29,0.29,0.29}
\definecolor {gray30}              {rgb}{0.30,0.30,0.30}
\definecolor {gray31}              {rgb}{0.31,0.31,0.31}
\definecolor {gray32}              {rgb}{0.32,0.32,0.32}
\definecolor {gray33}              {rgb}{0.33,0.33,0.33}
\definecolor {gray34}              {rgb}{0.34,0.34,0.34}
\definecolor {gray35}              {rgb}{0.35,0.35,0.35}
\definecolor {gray36}              {rgb}{0.36,0.36,0.36}
\definecolor {gray37}              {rgb}{0.37,0.37,0.37}
\definecolor {gray38}              {rgb}{0.38,0.38,0.38}
\definecolor {gray39}              {rgb}{0.39,0.39,0.39}
\definecolor {gray40}              {rgb}{0.40,0.40,0.40}
\definecolor {gray42}              {rgb}{0.42,0.42,0.42}
\definecolor {gray43}              {rgb}{0.43,0.43,0.43}
\definecolor {gray44}              {rgb}{0.44,0.44,0.44}
\definecolor {gray45}              {rgb}{0.45,0.45,0.45}
\definecolor {gray46}              {rgb}{0.46,0.46,0.46}
\definecolor {gray47}              {rgb}{0.47,0.47,0.47}
\definecolor {gray48}              {rgb}{0.48,0.48,0.48}
\definecolor {gray49}              {rgb}{0.49,0.49,0.49}
\definecolor {gray50}              {rgb}{0.50,0.50,0.50}
\definecolor {gray51}              {rgb}{0.51,0.51,0.51}
\definecolor {gray52}              {rgb}{0.52,0.52,0.52}
\definecolor {gray53}              {rgb}{0.53,0.53,0.53}
\definecolor {gray54}              {rgb}{0.54,0.54,0.54}
\definecolor {gray55}              {rgb}{0.55,0.55,0.55}
\definecolor {gray56}              {rgb}{0.56,0.56,0.56}
\definecolor {gray57}              {rgb}{0.57,0.57,0.57}
\definecolor {gray58}              {rgb}{0.58,0.58,0.58}
\definecolor {gray59}              {rgb}{0.59,0.59,0.59}
\definecolor {gray60}              {rgb}{0.60,0.60,0.60}
\definecolor {gray61}              {rgb}{0.61,0.61,0.61}
\definecolor {gray62}              {rgb}{0.62,0.62,0.62}
\definecolor {gray63}              {rgb}{0.63,0.63,0.63}
\definecolor {gray64}              {rgb}{0.64,0.64,0.64}
\definecolor {gray65}              {rgb}{0.65,0.65,0.65}
\definecolor {gray66}              {rgb}{0.66,0.66,0.66}
\definecolor {gray67}              {rgb}{0.67,0.67,0.67}
\definecolor {gray68}              {rgb}{0.68,0.68,0.68}
\definecolor {gray69}              {rgb}{0.69,0.69,0.69}
\definecolor {gray70}              {rgb}{0.70,0.70,0.70}
\definecolor {gray71}              {rgb}{0.71,0.71,0.71}
\definecolor {gray72}              {rgb}{0.72,0.72,0.72}
\definecolor {gray73}              {rgb}{0.73,0.73,0.73}
\definecolor {gray74}              {rgb}{0.74,0.74,0.74}
\definecolor {gray75}              {rgb}{0.75,0.75,0.75}
\definecolor {gray76}              {rgb}{0.76,0.76,0.76}
\definecolor {gray77}              {rgb}{0.77,0.77,0.77}
\definecolor {gray78}              {rgb}{0.78,0.78,0.78}
\definecolor {gray79}              {rgb}{0.79,0.79,0.79}
\definecolor {gray80}              {rgb}{0.80,0.80,0.80}
\definecolor {gray81}              {rgb}{0.81,0.81,0.81}
\definecolor {gray82}              {rgb}{0.82,0.82,0.82}
\definecolor {gray83}              {rgb}{0.83,0.83,0.83}
\definecolor {gray84}              {rgb}{0.84,0.84,0.84}
\definecolor {gray85}              {rgb}{0.85,0.85,0.85}
\definecolor {gray86}              {rgb}{0.86,0.86,0.86}
\definecolor {gray87}              {rgb}{0.87,0.87,0.87}
\definecolor {gray88}              {rgb}{0.88,0.88,0.88}
\definecolor {gray89}              {rgb}{0.89,0.89,0.89}
\definecolor {gray90}              {rgb}{0.90,0.90,0.90}
\definecolor {gray91}              {rgb}{0.91,0.91,0.91}
\definecolor {gray92}              {rgb}{0.92,0.92,0.92}
\definecolor {gray93}              {rgb}{0.93,0.93,0.93}
\definecolor {gray94}              {rgb}{0.94,0.94,0.94}
\definecolor {gray95}              {rgb}{0.95,0.95,0.95}
\definecolor {gray97}              {rgb}{0.97,0.97,0.97}
\definecolor {gray98}              {rgb}{0.98,0.98,0.98}
\definecolor {gray99}              {rgb}{0.99,0.99,0.99}
\definecolor {darkgrey}            {rgb}{0.66,0.66,0.66}

\newcommand{\resp}[1]{[resp.\ #1]}

\newcommand{\TODO}[1]{{}}

\newcommand{\ignore}[1]{}

\newcommand{\RSTODO}[1]{{\bf \textcolor{darkgreen}{{\fbox{RS TODO:} #1}}}}
\renewcommand{\RSTODO}[1]{}

\newcommand{\ignoreinshort}[1]{}
\newcommand{\ignoreinlong}[1]{{#1}}

\providecommand{\longversion}{true}
\ifthenelse{\equal{\longversion}{true}}
{%
    \renewcommand{\ignoreinshort}[1]{\textcolor{blue}{#1}}
    \newcommand{\ignoreinshortnc}[1]{{#1}}
    \renewcommand{\ignoreinlong}[1]{}
    \newcommand{\ignoreinlongnc}[1]{}
    
}%
{%
    \renewcommand{\ignoreinshort}[1]{}
    \newcommand{\ignoreinshortnc}[1]{}
    \renewcommand{\ignoreinlong}[1]{\textcolor{blue}{#1}}
    \newcommand{\ignoreinlongnc}[1]{{#1}}
    
}

\def\makenewenumerate#1#2{%
    \newcounter{cnt#1}
    \newenvironment{#1}%
    {\begin{list}{\makebox[0pt][r]{#2}}%
            {\setlength{\itemsep}{0pt}%
                \setlength{\parsep}{.2em}%
                \setlength{\leftmargin}{1.5em}%
                \setlength{\labelwidth}{.4em}%
                \usecounter{cnt#1}}}
            {\end{list}}}

\makenewenumerate{myenumerate}{\arabic{cntmyenumerate}.}

\makenewenumerate{renumerate}{\rm(\roman{cntrenumerate})}
\makenewenumerate{renumerateprime}{\rm(\roman{cntrenumerateprime}$'$)}
\makenewenumerate{renumeratesecond}{\rm(\roman{cntrenumeratesecond}$''$)}

\makenewenumerate{aenumerate}{({\it\alph{cntaenumerate}})}
\makenewenumerate{aenumerateprime}{({\it\alph{cntaenumerateprime}$'$})}
\makenewenumerate{aenumeratesecond}{({\it\alph{cntaenumeratesecond}$''$})}

\newcommand{\sref}[1]{\S{}\ref{#1}}

\newcommand{\set}[1]{\ensuremath{\{{#1}\}}\xspace}

\renewcommand{\iff}{\ensuremath{\leftrightarrow}\xspace}
\newcommand{\defas}{\ensuremath{\stackrel{\text{\scalebox{.7}{def}}}{=}}\xspace}

\newcommand\mysout{\bgroup \markoverwith{{-}}\ULon}
\newcommand\nosout{\bgroup \markoverwith{{ }}\ULon}
\definecolor{mygray}{rgb}{0.90,0.90,0.90}
\definecolor{mywhite}{rgb}{1.00,1.00,1.00}

\newcommand{\minimize}{\ensuremath{\mathsf{Minimize}}\xspace}

\newcommand{\satres}{\textsc{sat}\xspace}

\newcommand{\proptofol}{\ensuremath{{\cal B}2{\cal T}}\xspace}
\newcommand{\foltoprop}{\ensuremath{{\cal T}2{\cal B}}\xspace}

\newcommand{\atoms}[1]{\ensuremath{Atoms(#1)}\xspace}

\newcommand{\B}{\ensuremath{\mathcal{B}}\xspace}
\newcommand{\T}{\ensuremath{\mathcal{T}}\xspace}

\newcommand{\smttt}[1]{\ensuremath{\text{SMT}(#1)}\xspace}

\newcommand{\euf}{\ensuremath{\mathcal{EUF}}\xspace}

\newcommand{\eq}{\ensuremath{\mathcal{E}}\xspace}
\newcommand{\dl}{\ensuremath{\mathcal{DL}}\xspace}

\newcommand{\larat}{\ensuremath{\mathcal{LA}(\mathbb{Q})}\xspace}

\renewcommand{\larat}{\ensuremath{\mathcal{LRA}}\xspace}

\newcommand{\bv}{\ensuremath{\mathcal{BV}}\xspace}
\newcommand{\mem}{\ensuremath{\mathcal{AR}}\xspace}

\newcommand{\smtlarat}{\smttt{\larat}}

\newcommand{\Tmodels}{\models_{\T}}

\newcommand{\pmodels}{\models_p}

\newcommand{\Tlemmas}{\T-lemmas\xspace}

\newcommand{\mathsat}{\textsc{MathSAT}\xspace}

\newcommand{\mathsatfive}{\textsc{MathSAT5}\xspace}

\renewcommand{\TODO}[1]{\todo[inline,color=green!40]{{\small{TODO: #1}}}}

\renewcommand{\RSTODO}[1]{\todo[inline,color=green!40]{{\small{RS TODO: #1}}}}

\newcommand{\allalpha}{\ensuremath{\boldsymbol{\alpha}}\xspace}
\newcommand{\allalphaprime}{\ensuremath{\allalpha^\prime}\xspace}
\newcommand{\allalphasecond}{\ensuremath{\allalpha^{\prime\prime}}\xspace}

\newcommand{\vi}{\ensuremath{\varphi}}
\newcommand{\viprime}{\ensuremath{\varphi'}}

\renewcommand{\B}{\ensuremath{\mathbb{B}}\xspace}

\newcommand{\allsym}[1]{\ensuremath{{\underline{\boldsymbol#1}}}}
\renewcommand{\allsym}[1]{\ensuremath{{\boldsymbol#1}}}

\renewcommand{\allalpha}{\allsym{\alpha}}

\newcommand{\etaa}[1]{\ensuremath{\eta_{#1}[\allalpha]}}
\newcommand{\rhoa}[1]{\ensuremath{\rho_{#1}[\allalpha]}}

\newcommand{\mua}[1]{\ensuremath{\mu_{#1}[\allalpha]}}
\newcommand{\Ca}[1]{\ensuremath{C_{#1}[\allalpha]}}

\newcommand{\via}[1]{\ensuremath{\varphi_{#1}[\allalpha]}}

\renewcommand{\viprime}[1]{\ensuremath{\varphi_{#1}'}}
\newcommand{\viaprime}[1]{\ensuremath{\viprime{#1}[\allalpha]}}

\newcommand{\etaaprime}{\ensuremath{\eta^\prime[\allalphaprime]}}
\newcommand{\muaprime}[1]{\ensuremath{\mu_{#1}'[\allalphaprime]}}
\newcommand{\etaasecond}{\ensuremath{\eta^{\prime\prime}[\allalphasecond]}}

\newcommand{\viprimea}[1]{\ensuremath{\varphi_{#1}'[\allalpha]}}

\newcommand{\rhoprimea}[1]{\ensuremath{\rho_{#1}'[\allalpha]}}

\newcommand{\bequiv}{\ensuremath{\equiv_{\B}}}
\newcommand{\Tequiv}{\ensuremath{\equiv_{\T}}}

\newcommand{\theorycla}{\ensuremath{\{\Ca{1},\ldots,\Ca{K}\}}}

\renewcommand{\pmodels}{\models_{\B}}

\newcommand{\aone}{\ensuremath{(x \le 0)}\xspace}
\newcommand{\atwo}{\ensuremath{(x = 1)}\xspace}

\newcommand{\CTTA}[1]{\ensuremath{H}(#1)} %
\newcommand{\ITTA}[1]{\ensuremath{P}(#1)} %
\renewcommand{\CTTA}[1]{\ensuremath{H_{\allalpha}(#1)}} %
\renewcommand{\ITTA}[1]{\ensuremath{P_{\allalpha}(#1)}} %

\newcommand{\TLEMMAS}[1]{\ensuremath{Cl_{\allalpha}(#1)}}

\newcommand{\allalphaT}{\ensuremath{\allalpha_{\T}}}
\newcommand{\allalphaB}{\ensuremath{\allalpha_{\B}}}
\newcommand{\allalphared}{\ensuremath{\allalpha_1}}

\newcommand{\allalphablu}{\ensuremath{\allalpha_2}}
\renewcommand{\allalphared}{\ensuremath{\textcolor{red}{\allalpha_1}}}
\renewcommand{\allalphablu}{\ensuremath{\textcolor{blue}{\allalpha_2}}}
\newcommand{\etaaT}[1]{\ensuremath{{\eta_\T}_{#1}[\allalphaT]}}
\newcommand{\etaaB}[1]{\ensuremath{{\eta_\B}_{#1}[\allalphaT]}}
\newcommand{\rhoaT}[1]{\ensuremath{{\rho_\T}_{#1}[\allalphaT]}}
\newcommand{\rhoaB}[1]{\ensuremath{{\rho_\B}_{#1}[\allalphaB]}}
\newcommand{\rhoared}[1]{\ensuremath{\textcolor{red}{\rho_1}_{#1}[\allalphared]}}
\newcommand{\rhoablu}[1]{\ensuremath{\textcolor{blue}{\rho_2}_{#1}[\allalphablu]}}
\newcommand{\rhoprimeared}[1]{\ensuremath{\textcolor{red}{\rho^\prime_1}_{#1}[\allalphared]}}
\newcommand{\rhoprimeablu}[1]{\ensuremath{\textcolor{blue}{\rho^\prime_2}_{#1}[\allalphablu]}}
\newcommand{\etaared}{\ensuremath{\textcolor{red}{\eta_1}[\allalphared]}}
\newcommand{\etaablu}{\ensuremath{\textcolor{blue}{\eta_2}[\allalphablu]}}
\newcommand{\BCAT}[1]{\ensuremath{\mathit{BC}_{\allalphaT}(#1)}}
\newcommand{\BCARED}[1]{\ensuremath{\mathit{BC}_{\allalphared}(#1)}}
\newcommand{\BCABLU}[1]{\ensuremath{\mathit{BC}_{\allalphablu}(#1)}}
\newcommand{\TLEMMASANY}{\ensuremath{Cl_{\allalpha}}}
\newcommand{\TLEMMASANYPRIME}{\ensuremath{\TLEMMASANY^\prime}}

\newcommand{\CTAPRIMEANY}{\ensuremath{M_{\allalphaprime}}}

\newcommand{\Input}{\textbf{input:}\xspace}
\newcommand{\Output}{\textbf{output:}\xspace}

\newcommand{\projectedenum}{\ensuremath{\mathsf{ProjectedAllSMT}}}
\newcommand{\projectedenumtotal}{\ensuremath{\projectedenum_\mathsf{total}}}
\newcommand{\projectedenumpartial}{\ensuremath{\projectedenum_\mathsf{partial}}}
\newcommand{\project}{\ensuremath{\mathsf{Project}}}
\newcommand{\modelsapartial}{\ensuremath{{M_\allalpha}_p}}
\newcommand{\modelsaprimepartial}{\ensuremath{{M_{\allalphaprime}}_p}}
\newcommand{\lemmaspartial}{\ensuremath{{\TLEMMASANY}_p}}

\newcommand{\lemmastotal}{\ensuremath{{\TLEMMASANY}_t}}
\newcommand{\getTheoryAtoms}{\ensuremath{\mathsf{GetTheoryAtoms}}}
\newcommand{\enumerator}{\ensuremath{\mathsf{LemmaEnumerator}}}
\newcommand{\partitionAtoms}{\ensuremath{\mathsf{PartitionAtoms}}}
\newcommand{\SMTsolve}{\ensuremath{\mathsf{SMT.Solve}}}
\newcommand{\SMTgetAssignment}{\ensuremath{\mathsf{SMT.GetAssignment}}}
\newcommand{\SMTgetLemmas}{\ensuremath{\mathsf{SMT.GetLemmas}}}
\newcommand{\lemmaenumallsmt}{\ensuremath{\mathsf{LemmaEnumeratorAllSMT}}} \newcommand{\lemmaenumdc}{\ensuremath{\mathsf{LemmaEnumeratorD\&C}}}
\newcommand{\lemmaenumproj}{\ensuremath{\mathsf{LemmaEnumeratorWithProjection}}}
\newcommand{\lemmaenumpartitioning}{\ensuremath{\mathsf{LemmaEnumeratorWithPartitioning}}}

\newcommand{\expbaselinename}{{\textsf{Baseline}}}
\newcommand{\expdcname}{{\textsf{D\&C}}}
\newcommand{\expdcprojname}{{\expdcname+\textsf{Proj.}}}
\newcommand{\expdcprojpartname}{{\expdcprojname+\textsf{Part.}}}
\newcommand{\expbaseline}{\expbaselinename\xspace}
\newcommand{\expdc}{\expdcname\xspace}
\newcommand{\expdcproj}{\expdcprojname\xspace}
\newcommand{\expdcprojpart}{\expdcprojpartname\xspace}

\Crefname{property}{Property}{Properties}
\Crefname{property}{Property}{Properties}
\Crefname{appendix}{Appendix}{Appendices}
\Crefname{appendix}{Appendix}{Appendices}

\usepackage{color}

\begin{document}

\title{%
	Beyond Eager Encodings: A Theory-Agnostic Approach to Theory-Lemma Enumeration in SMT
}
\titlerunning{Beyond Eager Encodings: A \T-Agnostic Approach to \T-Lemma Enumeration}
\author{%
	Emanuele Civini\inst{1} \and
	Gabriele Masina\inst{1} \and
	Giuseppe Spallitta\inst{2} \and\\
	Roberto Sebastiani\inst{1}
}
\authorrunning{E. Civini et al.}
\institute{DISI, University of Trento, Trento, Italy
	\and Rice University, Houston, TX, USA
}
\maketitle              %

\renewcommand{\longversion}{true}
\ifthenelse{\equal{\longversion}{true}}
{%
	\renewcommand{\ignoreinshort}[1]{{\textcolor{midnightblue}{#1}}}
	\renewcommand{\ignoreinlong}[1]{}
	\specialcomment{IGNOREINSHORT}{\begingroup\color{midnightblue}}{\endgroup}
	\excludecomment{IGNOREINLONG}
}%
{%
	\renewcommand{\ignoreinshort}[1]{}
	\renewcommand{\ignoreinlong}[1]{#1}
	\excludecomment{IGNOREINSHORT}
	\specialcomment{IGNOREINLONG}{\begingroup}{\endgroup}
}
\excludecomment{IGNORE}

\begin{abstract}
	Lifting Boolean-reasoning techniques to the SMT level most often requires producing theory lemmas
that rule out theory-inconsistent truth assignments.
With standard SMT solving, it is common to ``lazily''
generate such lemmas on demand
during the search;
with some harder SMT-level tasks
---such as unsat-core extraction, MaxSMT, %
\T-OBDD or \T-SDD compilation---
it may be beneficial or even necessary to ``eagerly'' pre-compute all the needed theory
lemmas upfront. 
Whereas in principle ``classic'' eager SMT encodings could do the job,
they
are specific for very few and easy theories, %
they do not comply
with theory combination, and may produce lots of unnecessary lemmas. 

In this paper, we present theory-agnostic
methods for enumerating complete sets of
theory lemmas tailored to a given formula. 
Starting from AllSMT as a baseline approach, 
we propose improved lemma-enumeration techniques,
 including divide\&conquer,
 projected enumeration, and theory-driven partitioning,
which are highly parallelizable and which 
may drastically improve scalability. 
An experimental evaluation demonstrates that these techniques significantly enhance efficiency and enable the method to scale to substantially more complex instances.

\end{abstract}
\section{Introduction}
\paragraph{Context.}
Lifting Boolean-reasoning techniques to the SMT level most often
requires producing \T-valid clauses (theory lemmas or \T-lemmas)
that rule out truth assignments that propositionally satisfy the
formula but are \T-inconsistent.
     (E.g., given %
     $\vi\defas(x=0)\vee(x=1)$, the \T-lemma $\neg(x=0)\vee\neg(x=1)$ rules out the \T-inconsistent truth assignment $(x=0)\wedge(x=1)$.)
These lemmas prevent the Boolean engine
from (re)visiting \T-inconsistent regions of the search space.

With standard SMT solving, \T-lemmas are commonly generated %
``lazily''
on demand
during the search~\cite{barrettSatisfiabilityModuloTheories2021},
each time to rule out some \T-inconsistent truth assignment, generated by a SAT solver, propositionally satisfying the input formula.

With some harder SMT tasks, such as (minimal) unsat-core extraction~\cite{cimattiComputingSmallUnsatisfiable2011,GuthmannST16},
MaxSMT~\cite{cimattiModularApproachMaxSAT2013,FazekasBB18}, %
\T-OBDD or \T-SDD compilation~\cite{micheluttiCanonicalDecisionDiagrams2024}, it may be beneficial or even necessary to ``eagerly'' pre-compute upfront {\em a whole set of \T-lemmas} ruling out all the \T-inconsistent truth assignments %
propositionally  satisfying the formula.

Unfortunately, with the exception of our own %
work in~\cite{micheluttiCanonicalDecisionDiagrams2024} (see below), in the literature, there seems to be no theory-agnostic way to achieve
this task. 
Notice that,
whereas in principle ``classic'' eager SMT encodings
(e.g.\ \cite{VelevB03,StrichmanSB02,Strichman02}, see below) could do the job,
they
are specific for very few and easy theories (\euf,\dl,\larat),
they do not comply
with theory combination, and may produce lots of unnecessary lemmas.~%

\paragraph{Contributions.}
In this paper, we address this problem explicitly, and we present
efficient and theory-agnostic methods for enumerating complete sets of
\T-lemmas ruling out all \T-inconsistent truth assignments
propositionally satisfying a formula.
Starting from our AllSMT-based
technique in~\cite{micheluttiCanonicalDecisionDiagrams2024} as a baseline,  we introduce three orthogonal improvements that drastically enhance
scalability.
First, we show how to use partial-assignment enumeration to decompose
the Boolean search space into independent residual subproblems,
enabling a divide-and-conquer strategy that can be parallelized.
Second, we prove that projecting the enumeration on the \T-atoms 
still yields a complete set of \T-lemmas, reducing the dimension of
the search. Third, we show that when the \T-atoms of $\vi$ can be partitioned
into symbol-disjoint components, \T-lemmas can be enumerated independently
within each component. These improvements are
theory agnostic and can be combined into a single framework.

An experimental evaluation on all the problems in~\cite{micheluttiCanonicalDecisionDiagrams2024}  and on 
temporal planning
problems~\cite{valentiniTemporalPlanningIntermediate2020}
demonstrates that these techniques drastically enhance efficiency, enabling the scaling to substantially more complex instances.

\paragraph{Related work.}
With SMT via eager encodings
(e.g.\ \cite{VelevB03,StrichmanSB02,Strichman02}), an input
\T-formula $\vi$ is encoded into SAT by
\begin{enumerate*}[label=(\roman*)]
\item enumerating a
complete set of \T-lemmas ruling out all possible
\T-unsatisfiable truth assignment on $\atoms{\vi}$,
\item conjoining them to
$\vi$ and 
\item producing the Boolean abstraction of the result. 
\end{enumerate*}
For SMT solving, these techniques ultimately proved to be much less
efficient than the lazy approach~\cite{barrettSatisfiabilityModuloTheories2021}. 
\footnote{In the SMT literature, the expression ``eager
encoding'' is also used to denote another and quite successful approach to SMT solving
for finite-domain theories, like bit-vectors,
floating-point arithmetic
     or arithmetic over bounded integers,
which encodes domain variables via sequences of Boolean
atoms, and encodes functions and predicates into Boolean functions
accordingly.
We emphasize that this is {\em not} what we refer to here.
}

With the ``lemma-lifting'' approaches for (minimal) unsat-core extraction~\cite{cimattiComputingSmallUnsatisfiable2011,GuthmannST16} and for
MaxSMT~\cite{cimattiModularApproachMaxSAT2013,FazekasBB18} respectively,
an SMT solver is invoked on \T-inconsistent formulas to
produce the \T-lemmas necessary for detecting their
\T-inconsistency; these \T-lemmas are then conjoined to the formulas,
whose Boolean abstractions are then fed to a Boolean unsat-core
extractors or to a MaxSAT solver, respectively. 

Most relevantly, in \cite{micheluttiCanonicalDecisionDiagrams2024} we
have presented a novel theory-agnostic technique to leverage to SMT
level canonical Decision
Diagrams (DD), e.g.\ OBDDs~\cite{bryantGraphBasedAlgorithmsBoolean1986} and SDDs \cite{darwicheSDDNewCanonical2011}.
Importantly, unlike previous approaches, these decision diagrams are provably {\em \T-canonical}, i.e., \T-equivalent formulas are encoded
into the same DD.
The technique works by enumerating upfront a whole collection of
\T-lemmas ruling out all the \T-inconsistent truth assignments
propositionally satisfying the formula, conjoining them to the formula
and producing the decision diagrams out of it.
Unfortunately, the \T-lemma enumeration technique adopted there, which was based on total-assignment AllSMT,  turned out to be the main bottleneck of
the whole process.

\paragraph{Motivations and Goals.}
In Knowledge Compilation (KC) (see \cite{darwicheKnowledgeCompilationMap2002}), a propositional theory is compiled off-line into a target
language, which is then used on-line to answer a large number of queries in polytime. The key
motivation behind KC is to push as much of the computational overhead into
the offline compiling phase, which is amortized over all online queries.

The main goal of this work is to play as an enabling technique for our approach to lift current KC methods to the SMT setting. A first step %
was presented in~\cite{micheluttiCanonicalDecisionDiagrams2024} for OBDDs and SDDs, and we recently generalized it to d-DNNFs in~\cite{masina-sat26-ddnnf},
in such a way to preserve the polynomial cost of the queries from the Boolean
case
(see~\cite{masina-sat26-ddnnf}).
The need for an eager \T-lemma enumeration upfront is motivated in~\cite{micheluttiCanonicalDecisionDiagrams2024,masina-sat26-ddnnf}.

\paragraph{Content of the paper.}
In~\sref{sec:background} we introduce the necessary background on SMT and AllSMT enumeration, which is the basis of our methods. In~\sref{sec:enum-lemmas} we first recall the baseline method for \T-lemma enumeration from~\cite{micheluttiCanonicalDecisionDiagrams2024} (\sref{sec:baseline}), and then present three novel theory-agnostic techniques for enumerating \T-lemmas, including divide-and-conquer (\sref{sec:divide-and-conquer}), projected enumeration (\sref{sec:projection}), and theory-driven partitioning (\sref{sec:partitioning}).
In~\sref{sec:experiments} we experimentally evaluate our techniques on a broad set of SMT problems, showing a dramatic improvement in efficiency and scalability wrt.\ the baseline.
Finally, in~\sref{sec:conclusions}, we summarize our contributions and
discuss future research directions.

\section{Background}%
\label{sec:background}
\paragraph{Notation \& terminology.}
We assume the reader is familiar with the basic syntax, semantics, and results of propositional and first-order logics.
We adopt the following terminology and notation.

A propositional formula is either a Boolean constant ($\top$ or $\bot$ representing ``true'' and ``false'', respectively), a Boolean atom, or a formula built from propositional connectives ($\neg,\land,\lor,\to,\leftrightarrow$) over propositional formulas. Propositional satisfiability (SAT) is the problem of deciding the satisfiability of propositional formulas.
Satisfiability Modulo Theories (SMT) extends SAT to the context of first-order formulas modulo some background theory \T, which provides an intended interpretation for constant, function, and predicate symbols~\cite{barrettSatisfiabilityModuloTheories2021}. We restrict to quantifier-free formulas. 
A \T-formula is a combination of theory-specific atoms (\T-atoms) and
Boolean atoms (\B-atoms) via Boolean connectives
(``atoms'' denotes \T- and \B-atoms indifferently).
For instance, the theory of linear real arithmetic (\larat) provides
the standard interpretations of arithmetic operators ($+$, $-$,
$\cdot$) and relations ($=$, $\le$, $<$, $>$) over the reals. %
\larat-atoms are linear (in)equalities over
rational variables. 
An example of \larat-formula is $((x-y\le 3)\vee (x=z))$, where $x,y,z$ are \larat-variables, and $(x-y\le 3),(x=z)$ are \larat-atoms.
Other theories of interest include equalities (\eq), equalities with uninterpreted functions (\euf), bit-vectors (\bv), arrays (\mem), and combinations thereof.

\foltoprop{} is a bijective function (``theory to Boolean''),
called {\em Boolean 
abstraction},
which maps Boolean atoms into themselves,
 \T{}-atoms into fresh Boolean atoms, 
and is homomorphic wrt.\ Boolean connectives and set inclusion.
The function \proptofol (``Boolean to theory''), called {\em
  refinement},  is the inverse of \foltoprop. 
  (For instance
$\foltoprop(\{((x-y\le 3)\vee (x=z)) \}) =
\{(A_1\vee A_2) \}$, $A_1$ and $A_2$ being fresh
Boolean variables, and $\proptofol(\{\neg A_1, A_2\})=
\{\neg (x-y\le 3), (x=z)\}$.)%

The symbol
$\allalpha\defas\set{\alpha_i}_i$,
possibly with subscripts or superscripts, denotes
a set of atoms.
We represent truth assignments as conjunctions of literals.
We
denote by
$2^{\allalpha}$ the set of all total truth assignments on \allalpha.
The symbols
$\vi$, $\psi$ denote \T{}-formulas, and 
$\mu$, $\eta$, $\rho$ denote conjunctions of %
literals.
If $\foltoprop(\eta) \models \foltoprop(\vi)$, then we say that $\eta$
\emph{propositionally (or \B-)satisfies} $\vi$, written
$\eta\pmodels\vi$. 
(Notice that if $\eta\pmodels\vi$ then
$\eta\Tmodels\vi$, but not vice versa.)
The notion of propositional/\B- satisfiability,
entailment and validity follow straightforwardly.
When both $\vi\pmodels\psi$ and $\psi\pmodels\vi$, we say that
$\vi$ and $\psi$ are {\em propositionally/\B- equivalent}, written 
``$\vi\bequiv\psi$''.
When both $\vi\Tmodels\psi$ and $\psi\Tmodels\vi$, we say that
$\vi$ and $\psi$ are {\em \T-equivalent}, written 
``$\vi\Tequiv\psi$''.
 (Notice that if $\eta\bequiv\vi$ then
 $\eta\Tequiv\vi$, but not vice versa.)
We call a {\em \T-lemma} any \T-valid clause.

\paragraph{Basic definitions.}
We recall some definitions %
from~\cite{micheluttiCanonicalDecisionDiagrams2024}.

We denote by \via{} the fact that  \allalpha{} is a superset of the
set of atoms occurring in $\vi$ whose truth assignments we are interested
in. The fact that it is a superset is sometimes necessary for
comparing formulas with different sets of atoms:
$\via{}$ and $\viprime{}[\allalpha{}']$ can be compared only if they are both
considered as formulas on $\allalpha{}\cup\allalpha{}'$.
(E.g., in order to check that the propositional formulas $(A_1\vee A_2)\wedge (A_1\vee\neg A_2)$ and
$(A_1\vee A_3)\wedge (A_1\vee\neg A_3)$ are equivalent, we need
considering them as formulas %
on \set{A_1,A_2,A_3}.)

Given a set of atoms \allalpha{} and a \T-formula \via{}, we denote by
$\CTTA{\vi{}}\defas\set{\etaa{i}}_{i}$ and 
$\ITTA{\vi{}}\defas\set{\rhoa{j}}_{j}$ respectively~%
\footnote{%
versions of the Greek letters 
$\eta$ and $\rho$, respectively.}
the set of all
{\em \T-consistent} and that of all {\em \T-inconsistent} total truth
assignments on the set of atoms $\allalpha{}$ which
\B-satisfy $\vi$, i.e., s.t.
\begin{eqnarray}
  \label{eq:decomposition}
  \via{} \bequiv \bigvee_{\etaa{i}\in\CTTA{\vi}}\etaa{i} \vee \bigvee_{\rhoa{j}\in\ITTA{\vi}}\rhoa{j}.
\end{eqnarray}
  Given two \T-formulas \via{} and \viaprime{}, we have that:
  \begin{enumerate}[label=(\alph*)]%
  \label{prop:assignmentsets}
  \item $\CTTA{\vi{}}$, $\ITTA{\vi{}}$, $\CTTA{\neg\vi{}}$,
    $\ITTA{\neg\vi{}}$ are pairwise disjoint;
  \item $\CTTA{\vi{}} \cup \ITTA{\vi{}} \cup \CTTA{\neg\vi{}} \cup \ITTA{\neg\vi{}}=2^{\allalpha}$; 
  \item $\via{}\bequiv\viaprime{}$ if and only if
    $\CTTA{\vi{}}=\CTTA{\viprime{}}$ and $\ITTA{\vi{}}=\ITTA{\viprime{}}$;
  \item $\via{}\Tequiv\viaprime{}$ if and only if $\CTTA{\vi{}}=\CTTA{\viprime{}}$.
  \end{enumerate}

\begin{example}%
\label{ex:proposition1}
Let $\allalpha{}=\set{\alpha_1,\alpha_2}\defas\set{\aone,\atwo}$.
  Consider %
  $\via{1}\defas \aone \vee \atwo$ and
  $\via{2}\defas \neg\aone \iff \atwo$.  
We have that
$\via{1}\not\bequiv\via{2}$ and 
$\via{1}\Tequiv\via{2}$.
Then $\CTTA{\via{1}}=\CTTA{\via{2}}=\set{\eta_1,\eta_2}=
\set{\aone\wedge\neg\atwo,\neg\aone\wedge\atwo}$, 
$\ITTA{\via{1}}=\set{\aone\wedge\atwo}$ and
$\ITTA{\via{2}}=\emptyset$.

\end{example}

\begin{definition}%
  \label{def:ruleout-nobetas}
  We say that a set \theorycla{} of clauses %
  \textbf{rules out} a given set
  $\set{\rhoa{1},\ldots,\rhoa{M}}$ of 
  total truth assignments 
  if and only if,
  for every $\rhoa{j}$ in the set, there exists a $\Ca{l}$
  s.t.\ $\rhoa{j}\pmodels \neg \Ca{l}$, that is, if and only if
  \begin{displaymath}
    \bigvee_{j=1}^M\rhoa{j}\wedge\bigwedge_{l=1}^K\Ca{l}\bequiv\bot.
  \end{displaymath}
\end{definition}

Given \allalpha{} and some \T-formula \via{}, we denote
as $\TLEMMAS{\vi}$ any function which returns a set 
of \T-lemmas \theorycla{} which rules out \ITTA{\vi}.

\paragraph*{AllSMT and Projected AllSMT.}%
\label{sec:all-smt}
Given a \T-formula \via{}, AllSMT is the problem of enumerating a set of \T-satisfiable truth assignments covering all $\etaa{}\in\CTTA{\vi{}}$~\cite{lahiriSMTTechniquesFast2006}.
We call \emph{total AllSMT} the task of enumerating all \T-satisfiable \emph{total} truth assignments $\etaa{}\in\CTTA{\vi{}}$. We call \emph{partial AllSMT} the task of enumerating a set of \T-satisfiable \emph{partial} truth assignments $\mua{}$ such that
$\mua{}\pmodels\vi{}$ and every $\etaa{}\in\CTTA{\vi{}}$ is an extension of some $\mua{}$.
\begin{algorithm}[t]\caption{\projectedenum(\via{}, \allalphaprime{})}\label{alg:proj-all-smt}
\Input $\via{}$: a \T-formula, $\allalphaprime{}$: set of atoms to project on, s.t.\ $\allalphaprime\subseteq\allalpha{}$\\
\Output $\CTAPRIMEANY$: set of partial assignments $\muaprime{}$ for projected partial AllSMT\\
\phantom{\Output}$\TLEMMASANY$: set of \T-lemmas learned during the enumeration
\begin{algorithmic}[1]
\State $\CTAPRIMEANY{}\gets\emptyset$\label{line:emptyset}
\While {$\SMTsolve(\via{})=\satres$}\label{line:while}
\State $\etaa{} \gets \SMTgetAssignment()$ \Comment{find a \T\!\!\!-sat truth assignment on $\allalpha{}$ for \vi{}}\label{line:getassignment}
\State $\mua{}\gets\minimize(\etaa{}, \allalphaprime{}, \via{})$\Comment{minimize $\etaa{}$ wrt.\ \allalphaprime{}}\label{line:minimize}
\State $\muaprime{}\gets \project(\mua{}, \allalphaprime{})$ \Comment{keep literals on \allalphaprime{} only}\label{line:project}
\State $\CTAPRIMEANY{} \gets \CTAPRIMEANY{} \cup \set{\muaprime{}}$\label{line:add}
\State $\vi{} \gets \vi{} \wedge \neg \muaprime{}$\Comment{block $\muaprime{}$}\label{line:block}
\EndWhile\label{line:endwhile}
\State $\TLEMMASANY \gets \SMTgetLemmas()$\Comment{get the learned \T-lemmas}\label{line:getlemmas}
\State \Return $\CTAPRIMEANY{}, \TLEMMASANY$
\end{algorithmic}
\end{algorithm}

Sometimes, we are interested in enumerating only the truth assignments
on a subset $\allalphaprime\subseteq\allalpha$ of ``relevant'' atoms, which we call \emph{projection atoms}.
Let $\allalphasecond\defas\allalpha\setminus\allalphaprime$.
Following~\cite{lahiriSMTTechniquesFast2006}, we call \emph{projected total} \resp{\emph{partial}} \emph{AllSMT} the task of
 enumerating a set of \T-satisfiable \emph{total} \resp{\emph{partial}} truth assignments on $\allalphaprime{}$, $\muaprime{}$, such that:
\begin{enumerate}[label=(\roman*)] %
  \item for every \muaprime{} there exists a total truth assignment
$\etaasecond{}$ on $\allalphasecond$ s.t.\ $\mua{}\defas
\muaprime{}\wedge\etaasecond{}$ is such that $\mua{}\pmodels\vi{}$ and $\mua{}$ is \T-satisfiable;\label{item:projected-partial-1}
\item every $\etaa{} \in \CTTA{\vi}$ is an extension of some
    $\muaprime{}$.\label{item:projected-partial-2}
\end{enumerate}
Note that if $\allalphaprime=\allalpha$, %
the projected versions coincide with the non-projected ones.

\Cref{alg:proj-all-smt} describes a procedure for computing projected partial AllSMT, implemented, e.g., in \mathsat~\cite{mathsat5_tacas13}; the procedure can be easily adapted to compute the other versions of AllSMT, as we detail in the following.
In particular, we denote by \projectedenumtotal{} the version of \Cref{alg:proj-all-smt} for projected total AllSMT, and by \projectedenumpartial{} the version for projected partial AllSMT.

The algorithm is based on the blocking-clause approach, which iteratively checks for the satisfiability of \via{} conjoined with the negation of the previously found truth assignments, until it becomes unsatisfiable (lines~\ref{line:while}-\ref{line:endwhile}).

In each iteration, it first finds a \T-satisfiable total truth assignment $\etaa{}$ on all the atoms $\allalpha$, which propositionally satisfies \via{} (line~\ref{line:getassignment}). This guarantees that all its sub-assignments are also \T-satisfiable. 

The assignment is then minimized, by removing literals over
$\allalphaprime$ from $\etaa{}\defas\etaaprime{}\wedge\etaasecond$, to find a partial assignment
$\mua{}\defas
\muaprime{}\wedge\etaasecond{}\subseteq\etaa{}$ which still satisfies \via{} (line~\ref{line:minimize}).
For total (projected) AllSMT, this step is omitted and $\mua{}=\etaa{}$.
Then, the assignment is projected on $\allalphaprime$ (line~\ref{line:project}), obtaining $\muaprime{}$ satisfying \Cref{item:projected-partial-1}. 
If $\allalphaprime=\allalpha$, then $\muaprime{}=\mua{}$, and
projected AllSMT reduces to AllSMT.

The projected assignment $\muaprime{}$ is then added to the set
$\CTAPRIMEANY$ (line~\ref{line:add}), and it is blocked by conjoining
the \emph{blocking clause} $\neg\muaprime{}$ to \via{}
(line~\ref{line:block}).%
~%
(In actual implementations, $\neg\muaprime{}$ may be fed to the
conflict-analysis mechanism of the SAT solver, and the output
conflict clause be used as blocking  clause
\cite{lahiriSMTTechniquesFast2006}; in \Cref{alg:proj-all-smt} we omit
this step to simplify the narration.)
The loop terminates when \via{} becomes unsatisfiable, which guarantees \Cref{item:projected-partial-2}.

Finally, we return both the set of enumerated assignments $\CTAPRIMEANY$, and also the set of \T-lemmas $\TLEMMASANY$ (line~\ref{line:getlemmas}) learned by the SMT solver during the enumeration, i.e., during all the calls to $\SMTsolve{}$ (line~\ref{line:while}), since we will use this algorithm as a subroutine for enumerating \T-lemmas in~\sref{sec:enum-lemmas}.

We remark that $\T$-lemmas are typically much
smaller than the blocking clauses negating partial or total assignments,
as the former negate conflicting cores of $\T$-atoms;
also, a \T-satisfiable partial assignment only implies that {\em
one} of its total extension is \T-satisfiable, whereas a
\T-unsatisfiable one implies that {\em all} its total extensions
are \T-unsatisfiable.
Thus, whereas a \T-lemma usually rules out a large amount of
elements of \ITTA{\vi}, a blocking clause rules out only one or
relatively few elements of \CTTA{\vi}.
Thus, unless \CTTA{\vi} is drastically smaller than \ITTA{\vi}, typically in AllSMT most of the search effort is devoted to %
rule out the assignments in \CTTA{\vi}.

We also remark that \emph{partial} AllSMT may be dramatically faster than total AllSMT, since a single \T-satisfiable partial assignment represents up to $2^k$ total ones, where $k$ is the number of unassigned atoms.
For the same reason, \emph{projected} AllSMT may be dramatically faster than non-projected AllSMT, as a single \T-satisfiable assignment on $\allalphaprime$ represents up to $2^k$ \T-satisfiable assignments on $\allalpha$.

As a final note, we also remark that the results presented in the rest
of the paper also hold for projected AllSMT algorithms not based on blocking clauses, e.g.,~\cite{spallittaDisjointProjectedEnumeration2025}. For simplicity, we carry out the discussion taking \Cref{alg:proj-all-smt} as a reference.

\section{\T-Lemma Enumeration Techniques}%
\label{sec:enum-lemmas}
We present four theory-agnostic methods for enumerating a set \TLEMMAS{\vi} of \T-lemmas to rule out \ITTA{\vi} (see \Cref{def:ruleout-nobetas}). 
In~\sref{sec:baseline} we describe the method used in~\cite{micheluttiCanonicalDecisionDiagrams2024}, which we will use as a baseline.
The method is simple to implement, but relies on basic, total AllSMT enumeration, which is often intractable on real-world instances due to the combinatorial explosion of total assignments.
To address this issue, we introduce three orthogonal improvements that
drastically enhance scalability.
First, in~\sref{sec:divide-and-conquer} we describe a divide-and-conquer strategy that can be naturally parallelized. Second, in~\sref{sec:projection} we prove that projecting the enumeration onto the \T-atoms guarantees finding a complete set of \T-lemmas, thereby enabling more efficient enumeration.
Finally, in~\sref{sec:partitioning} we prove that when the atoms of \vi{} can be partitioned into disjoint sets, \T-lemmas can be enumerated for each partition separately. This results in a dramatic improvement in the efficiency of enumeration.
Notably, these three improvements are orthogonal and can be combined.
For the sake of compactness, all the
  proofs are %
deferred to \Cref{app:proofs}.%

\subsection{Enumerating \T-lemmas via AllSMT}%
\label{sec:baseline}
\begin{algorithm}[t]
\caption{$\lemmaenumallsmt(\via{}, \allalphaprime)$}\label{alg:lemmas:allsmt} 
\Input $\via{}$: a \T-formula, $\allalphaprime$: a set of atoms to project on (default: $\allalpha$) \\
\Output $\TLEMMASANY{}$: a set of \T-lemmas
\begin{algorithmic}[1]
\State $\_, \TLEMMASANY{} \gets \projectedenumtotal(\via{}, \allalphaprime)$\Comment{Keep only lemmas}
\State \Return $\TLEMMASANY{}$
\end{algorithmic}
\end{algorithm}
The baseline approach for $\mathcal{T}$-lemma enumeration, first presented in~\cite{micheluttiCanonicalDecisionDiagrams2024}, is outlined in \Cref{alg:lemmas:allsmt} and consists of performing a call to \emph{total, non-projected AllSMT} on $\via{}$, using the procedure in \Cref{alg:proj-all-smt} with $\allalphaprime\defas\allalpha$ and no minimization.
(In \Cref{alg:lemmas:allsmt} we also consider the case with $\allalphaprime\subset\allalpha$
  because it will be used in \sref{sec:projection} and \sref{sec:partitioning}.)

The \T-lemmas learned by the SMT solver during such enumeration %
rule out all \T-inconsistent truth assignments in $\ITTA{\vi}$.
In fact, the algorithm terminates when
$\vi\wedge\bigwedge_{\eta_i\in\CTTA{\vi}} {\neg\eta_i}$ is found to be
\T-unsatisfiable, that is, when the solver has found a set of \T-lemmas \TLEMMASANY{} s.t.\
$\vi\wedge\bigwedge_{\eta_i\in\CTTA{\vi}} {\neg\eta_i}\wedge\bigwedge_{C\in\TLEMMASANY{}}{C}\bequiv\bot$.
By construction, this guarantees that $\TLEMMASANY{}$ rules out all
\T-unsatisfiable truth assignments in $\ITTA{\vi}$~\cite{micheluttiCanonicalDecisionDiagrams2024}.

It is important to note that we must enumerate \emph{total} truth assignments. 
Indeed, enumerating \emph{partial} truth assignments may produce a set
of $\mathcal{T}$-lemmas that does not rule out $\ITTA{\vi}$ completely. 
This is because the definition of partial AllSMT (see~\sref{sec:all-smt}) only requires that, for each partial assignment, all its total extensions \emph{propositionally} satisfy $\vi$, but only one extension needs to be \T-consistent.
The following example illustrates why this can lead to missing \T-lemmas.

\begin{example}%
\label{ex:partial-miss-lemma}
Let $\via{}\defas(x=0)\lor(x=1)$ with
$\allalpha\defas\set{(x=0),(x=1)}$. Consider a partial AllSMT
disjoint enumeration as performed by \Cref{alg:proj-all-smt} with
$\allalphaprime=\allalpha$ and minimization enabled.
A possible run may proceed as follows.

The SMT solver sets
$\neg(x=0)$, so that unit propagation forces $(x=1)$. The resulting
total assignment $\etaa{}\defas\neg(x=0)\wedge(x=1)$ is \T-consistent
(line~\ref{line:getassignment}). Minimization then drops $\neg(x=0)$
and returns the partial assignment $\mua{1}\defas(x=1)$
(line~\ref{line:minimize}). Notice that no \T-conflict is exposed, and
no \T-lemma is generated. The enumerator blocks the returned cube by
adding $\neg(x=1)$ (line~\ref{line:block}). On the next iteration,
$\neg(x=1)$ is enforced, and unit propagation forces $(x=0)$.
The resulting total assignment
$\etaa{2}\defas(x=0)\wedge\neg(x=1)$ is again \T-consistent
(line~\ref{line:getassignment}). Minimization
(line~\ref{line:minimize}) cannot reduce \etaa{2} because of the
blocking clause $\neg(x=1)$, 
and $\etaa{2}$ is then blocked
by adding $\neg(x=0)\vee(x=1)$ (line~\ref{line:block}), terminating the search.

Along this run, no \T-conflict is encountered, hence no \T-lemma is generated. This is because the SMT solver never explored the \T-inconsistent total assignment $(x=0)\wedge(x=1)$ which belongs to $\ITTA{\vi}$, and no \T-lemma is generated to rule it out.

\end{example}

This example shows that naive partial AllSMT enumeration \emph{does not} %
reliably compute $\TLEMMAS{\vi}$. Enumerating total assignments, however, is typically infeasible on real-world instances due to prohibitive computational time and memory blow-ups. In fact, modern enumeration research emphasizes the importance of producing short assignments to scan the search space more efficiently (see e.g.,~\cite{friedEntailingGeneralizationBoosts2024,masinaCNFConversionSAT2025,spallittaDisjointProjectedEnumeration2025}). These observations motivate the need for methods that produce $\TLEMMAS{\vi}$ without exhaustive enumeration, while still leveraging compact partial assignments.

\begin{remark}%
  \label{rem:projection-partitioning}  
  For the case $\allalphaprime\subset\allalpha$, the algorithm, in general, does not find a complete set of \T-lemmas ruling out \ITTA{\vi}. In~\sref{sec:projection} and~\sref{sec:partitioning}, however, we will show some special cases for $\allalphaprime$ where the projected version is used without affecting completeness.
\end{remark}

\subsection{Divide and Conquer (and Parallelize)}%
\label{sec:divide-and-conquer}
\begin{algorithm}[t]
\caption{$\lemmaenumdc(\via{}, \allalphaprime)$}\label{alg:lemmas:divide-and-conquer}
\Input $\via{}$: a \T-formula, $\allalphaprime$: a set of atoms to project on (default: $\allalpha$) \\
\Output $\TLEMMASANY{}$: a set of \T-lemmas
\begin{algorithmic}[1]
\State $\modelsaprimepartial,\lemmaspartial \gets \projectedenumpartial(\via{}, \allalphaprime)$\label{line:dc:partial} 
\State $\TLEMMASANY \gets \lemmaspartial$\label{line:dc:lemmaspartial}
\ForAll{$\mu \in \modelsaprimepartial$}  \Comment{Parallel loop}\label{line:dc:parallel-for}
    \State $\_, \lemmastotal \gets \projectedenumtotal((\vi \wedge \mu \wedge \bigwedge_{C\in{\lemmaspartial}} C)[\allalpha], \allalphaprime)$\label{line:dc:total}
    \State $\TLEMMASANY \gets \TLEMMASANY \cup \lemmastotal$\label{line:dc:extend-lemmas}
\EndFor\label{line:dc:parallel-for-end}
\State \Return $\TLEMMASANY$\label{line:dc:return}
\end{algorithmic}
\end{algorithm}
Although %
partial enumeration alone is not sufficient to find a set of \T-lemmas that rules out $\ITTA{\vi}$, it provides a useful decomposition of the Boolean search space into a family of residual, \T-satisfiable subproblems.
This decomposition is inspired by the cube-and-conquer approach in SAT~\cite{heuleCubeConquerGuiding2012}, recently used also for model counting~\cite{xuEmbarrassinglyParallelModel2025}.

To exploit this decomposition, we propose a divide-and-conquer
approach to enumerate \T-lemmas, outlined in
\Cref{alg:lemmas:divide-and-conquer}.
We consider first the case where $\allalphaprime=\allalpha$.

First, the algorithm performs a {\em partial} AllSMT enumeration on $\via{}$, obtaining a set of partial truth assignments $\modelsapartial=\{\mua{1},\dots,\mua{k}\}$, and also a set of \T-lemmas $\lemmaspartial$ learned during the enumeration (line~\ref{line:dc:partial}). 

Then, for each $\mua{i}\in\modelsapartial$
(lines~\ref{line:dc:parallel-for}--\ref{line:dc:parallel-for-end}),
the algorithm performs a {\em total} AllSMT enumeration on $\vi \land
\mua{i}$, conjoining also the previously learned \T-lemmas $\lemmaspartial$,
to prevent the SMT solver from rediscovering them in the residual subproblems (line~\ref{line:dc:total}). After each iteration, we save the lemmas learned by total AllSMT.
Overall, rather than paying for one large total AllSMT run
on $\via{}$, we pay for one cheaper decomposition step and for many
smaller total AllSMT runs on restricted residual formulas,  whose
overall cost is typically lower than one monolithic total enumeration on $\via{}$.

Moreover, the residual enumeration tasks (line~\ref{line:dc:total}) can be executed in parallel on multi-core machines, e.g., by assigning each $\mu_i$ to a process and finally merging the $\T$-lemmas produced by the processes.
In our architecture, processes execute in isolation, and distinct processes may independently discover identical or redundant lemmas.
In practice, in~\sref{sec:experiments} we show empirically that the number of lemmas found is comparable to that of the AllSMT baseline.

The correctness of this approach follows from the fact that overall, in the different iterations, all the total truth assignments are explored, thus for each $\rhoa{} \in \ITTA{\vi}$, a $\T$-lemma is eventually generated that rules it out.

We remark that, while our algorithm guarantees that the $\mua{i}$ are
disjoint, \T-satisfiable, and propositionally satisfy $\via{}$, the
latter fact is not a strict requirement for the correctness of the approach.
Indeed, this is only one of the possible decompositions of the search space, and any family of partial assignments that covers all propositional assignments of $\vi$ would suffice. This opens up the possibility of designing alternative heuristics to decompose the search space.

For the case $\allalphaprime\subset\allalpha$, \Cref{rem:projection-partitioning} applies for \Cref{alg:lemmas:divide-and-conquer} as well.

\subsection{Projection on Theory-Atoms}%
\label{sec:projection}
\begin{algorithm}[t]
\caption{$\lemmaenumproj(\via{}, \enumerator)$}\label{alg:lemmas:projection} 
\Input $\via{}$: a \T-formula, $\enumerator$: a \T-lemma enumerator \\
\Output $\TLEMMASANY{}$: a set of \T-lemmas
\begin{algorithmic}[1]
\State $\allalphaT \gets \getTheoryAtoms(\allalpha)$ \Comment{Extract theory atoms from $\allalpha$}
\State $\TLEMMASANY{} \gets \enumerator(\via{}, \allalphaT)$
\State \Return $\TLEMMASANY{}$
\end{algorithmic}
\end{algorithm}
Enumerating all total truth assignments, even when parallelized, is often infeasible on real-world instances. In this section, we spotlight that this is often unnecessary to find a complete set of \T-lemmas, starting from the following observation.

Truth assignments explored during an AllSMT enumeration, whether via the classic or the divide-and-conquer approach, consist of two components: an assignment to \T-atoms and an assignment to \B-atoms.
We note that the \T-inconsistency of a truth assignment is determined exclusively by its projection onto the \T-atoms.
Consequently, a standard AllSMT procedure may inefficiently enumerate multiple truth assignments that share an identical \T-consistent projection but differ in their Boolean assignments. This redundancy is computationally wasteful, as variations in the Boolean assignment do not generate new \T-lemmas.

We can address this inefficiency by performing an {AllSMT
  enumeration \emph{projected on the \T-atoms of \vi}. The correctness of this approach follows from the following theorem (proof in
  \Cref{app:proofs:teo:lemmasout-exists}%
  ).
\begin{theorem}%
    \label{teo:lemmasout-exists}
    Let \via{} be a \T-formula.
    Let \allalphaB{} and \allalphaT{} be the sets of Boolean and theory atoms in \allalpha{}, respectively.

    Let $\BCAT{\vi}\defas\set{\neg\etaaT{}\mathrm{\ s.t.\ }\etaaT{}\wedge\etaaB{}\in\CTTA{\vi}}$, and let \TLEMMASANY{} be a set of \T-lemmas s.t.\ $\BCAT{\vi}\cup\TLEMMASANY{}$ rules out $\ITTA{\vi}$.
    
    Then $\TLEMMASANY{}$ rules out $\ITTA{\vi}$.
\end{theorem}

This leads to the enumeration procedure outlined in \Cref{alg:lemmas:projection}. This procedure invokes an AllSMT-based \T-lemma enumerator (e.g., the one in \Cref{alg:lemmas:allsmt} or \Cref{alg:lemmas:divide-and-conquer}), projecting the enumeration onto \T-atoms only.
If the lemma enumerator is based on \Cref{alg:proj-all-smt}, which terminates when $\vi\wedge\bigwedge_i{\neg\etaaT{i}}\wedge\bigwedge_{C\in\TLEMMASANY{}}{C}\bequiv\bot$, it follows from \Cref{teo:lemmasout-exists} that the returned set \TLEMMASANY{} of \T-lemmas rules out \ITTA{\vi}.

\begin{remark}%
\label{rem:projection-efficiency}
Note that the projected assignments can be \emph{exponentially fewer} than the assignments on the full set of atoms. Also, the blocking clauses added by \Cref{alg:proj-all-smt} (line~\ref{line:block}) are much shorter, as they only contain \T-atoms. This can lead to a drastic reduction in both enumeration time and memory usage.
\end{remark}

\subsection{Theory-driven partitioning}%
\label{sec:partitioning}

\begin{algorithm}[t]
\caption{$\lemmaenumpartitioning(\via{}, \enumerator)$}\label{alg:lemmas:partitioning} 
\Input $\via{}$: a \T-formula, $\enumerator$: a \T-lemma enumerator \\
\Output $\TLEMMASANY{}$: a set of \T-lemmas
\begin{algorithmic}[1]
\State $\TLEMMASANY \gets \emptyset$\label{line:partitioning:init}
\State $\set{\allalpha_1, \ldots, \allalpha_n} \gets \partitionAtoms(\allalpha)$ \Comment{Partition atoms as in \Cref{prop:partitioning}}\label{line:partitioning:partition}
\ForAll{$i \in [1..n]$}\label{line:partitioning:loop}
    \State $\TLEMMASANY{}_i \gets \enumerator((\vi\wedge\bigwedge_{C\in\TLEMMASANY} C)[\allalpha], \allalpha_i)$\label{line:partitioning:enumeration}
    \State $\TLEMMASANY \gets \TLEMMASANY \cup \TLEMMASANY{}_i$
\EndFor\label{line:partitioning:end}
\State \Return $\TLEMMASANY$\label{line:partitioning:return}
\end{algorithmic}
\end{algorithm}

We now present a further improvement for \T-lemmas enumeration, for the
case where the set of atoms of 
$\via{}$ can be partitioned into two or more disjoint sets wrt.\ variables and uninterpreted symbols.
We first recall the following property, which characterizes the partitioning of atoms that we will use in the following results.
\begin{property}%
    \label{prop:partitioning}
    Let $\allalpha{}\defas\allalphared{}\cup\allalphablu{}$ be a set of atoms s.t.\ no atom in \allalphared{} shares variables or uninterpreted symbols with any atom in \allalphablu{}.
    Then, for any truth assignment
    $\etaa{}\defas\etaared{}\wedge\etaablu{}$,
    $\etaa{}$ is \T-satisfiable iff both $\etaared{}$ and $\etaablu{}$ are \T-satisfiable. 
\end{property}

Consequently, if the set of atoms of $\vi$ can be partitioned according to the conditions of \Cref{prop:partitioning}, then in order to rule out a \T-inconsistent truth assignment $\rhoa{}\defas\rhoared{}\wedge\rhoablu{}\in\ITTA{\vi}$, a \T-lemma must rule out $\rhoared{}$ or $\rhoablu{}$.

This suggests a decomposition strategy: an AllSMT enumeration projected on \allalphared{} will find a set of \T-lemmas that rule out all $\rhoa{}\in\ITTA{\vi}$ such that $\rhoared{}$ is \T-inconsistent.
Dually, an AllSMT enumeration projected on \allalphablu{}, will find a set of \T-lemmas that rule out all $\rhoa{}\in\ITTA{\vi}$ such that $\rhoablu{}$ is \T-inconsistent.
Thus, the enumeration of \T-lemmas can be decomposed into two separate enumerations.

This intuition is formalized in the following theorem (proof in
  \Cref{app:proofs:teo:lemmasout-partitioning}%
  ).

\begin{theorem}%
    \label{teo:lemmasout-partitioning}
    Let \via{} be a \T-formula, s.t.\ $\allalpha{}\defas\allalphared{}\cup\allalphablu{}$ and no atom in \allalphared{} shares variables or uninterpreted symbols with any atom in \allalphablu{}.

    Let $\BCARED{\vi}\defas\set{\neg \etaared{} \mathrm{\ s.t.\ } \etaared\wedge\etaablu{}\in\CTTA{\vi}}$, and let \TLEMMASANY{} be a set of \T-lemmas s.t.\ $\BCARED{\vi}\cup\TLEMMASANY{}$ rules out $\ITTA{\vi}$.
    
    Let $\viprimea{}\defas\vi\wedge\bigwedge_{C\in\TLEMMASANY{}} C$.

    Let $\BCABLU{\viprime{}}\defas\set{\neg \etaablu \mathrm{\ s.t.\ }\etaared\wedge\etaablu{}\in\CTTA{\viprime{}}}$, and let \TLEMMASANYPRIME{} be a set of \T-lemmas s.t.\ $\BCABLU{\viprime{}}\cup\TLEMMASANYPRIME{}$ rules out $\ITTA{\viprime{}}$.

    Then $\TLEMMASANY{}\cup\TLEMMASANYPRIME{}$ rules out $\ITTA{\vi}$.

\end{theorem}

Note that this result can be easily generalized to the case where the
set of atoms can be partitioned into more than two sets. Indeed, if $\allalphablu{}$ can be further partitioned, we can apply the same
reasoning recursively to it.

This result naturally leads to the enumeration strategy in \Cref{alg:lemmas:partitioning}.
The algorithm first partitions the set of atoms of $\via{}$ as in \Cref{prop:partitioning}, obtaining $\set{\allalpha_1, \ldots, \allalpha_n}$ (line~\ref{line:partitioning:partition}).
This step can be performed efficiently, e.g., using a union-find data
structure to compute the connected components of the graph of
atoms, where two atoms are directly connected if they share a variable or an uninterpreted symbol.
Note that, in principle, every \B-atom can be considered as belonging to a separate, singleton partition.
In practice, we assume that all the \B-atoms belong to the same partition, and can be handled via projection as discussed in \sref{sec:projection}. 

Then, for each partition $\allalpha_i$, the algorithm performs an
AllSMT-based \T-lemma enumeration projected on $\allalpha_i$
(line~\ref{line:partitioning:enumeration}), and accumulates the
enumerated \T-lemmas in \TLEMMASANY{}
(line~\ref{line:partitioning:enumeration}). Note that, as in
\Cref{teo:lemmasout-partitioning}, we conjoin the returned \T-lemmas
to the input formula in each enumeration, so as to rule out \T-inconsistent assignments that have already been ruled out by previous enumerations.
Finally, we return the accumulated set \TLEMMASANY{} of \T-lemmas (line~\ref{line:partitioning:return}), which, by \Cref{teo:lemmasout-partitioning}, rules out $\ITTA{\vi}$.

Similar considerations as in \Cref{rem:projection-efficiency} apply in
this case. Indeed, if $n>1$ then the improvement could be dramatic, as
we decompose a single total enumeration step into $n$ projected enumerations on \emph{subsets} of atoms.

\section{Experimental Evaluation}
We evaluate the performance of the \T-lemma enumeration techniques
of~\sref{sec:enum-lemmas} in terms of \emph{time} required
to enumerate \T-lemmas. We also analyze the \emph{number} of
enumerated \T-lemmas, and their median \emph{size}. We consider the following enumerators:
\begin{itemize}
  \item \expbaseline, i.e., standard AllSMT as in~\cite{micheluttiCanonicalDecisionDiagrams2024} (\Cref{alg:lemmas:allsmt}, \sref{sec:baseline}); 
  \item \expdc, i.e., divide-and-conquer AllSMT (\Cref{alg:lemmas:divide-and-conquer}, \sref{sec:divide-and-conquer});
  \item \expdcproj, i.e., \expdc{} plus projection on \T-atoms (\Cref{alg:lemmas:projection}, \sref{sec:projection});
  \item \expdcprojpart, i.e., \expdcproj{} plus partitioning (\Cref{alg:lemmas:partitioning}, \sref{sec:partitioning}).
\end{itemize}

\paragraph{Reproducibility.}
All the experiments have been executed on three equivalent machines with an Intel(R) Xeon(R) Gold 6238R @ 2.20GHz CPU and 128GB of RAM.
Divide-and-conquer-based enumerators use 45 cores for parallelization.
For each instance, we set a timeout of 3600s.
We use PySMT~\cite{garioPySMTSolveragnosticLibrary2015} to manipulate SMT formulas and interface with solvers, and \mathsatfive~\cite{mathsat5_tacas13} as AllSMT solver, as it implements the enumeration procedure described in \Cref{alg:proj-all-smt}.
For finding short partial assignments for non-CNF formulas we use the technique in~\cite{masinaCNFConversionSAT2025}.
The code for the algorithms~\footnote{\url{https://github.com/ecivini/tlemmas-enumeration}} and for reproducing the experiments~\footnote{\url{https://github.com/ecivini/tlemmas-enumeration-testbench}} are publicly available.

\paragraph{Benchmarks.}

We considered two sets of benchmarks.
The first set consists of all the 450 \smtlarat{} synthetic instances from~\cite{micheluttiCanonicalDecisionDiagrams2024}.
These problems were generated randomly by nesting Boolean operators up to a fixed depth varying from 4 to 8, involving 10 Boolean variables and 10 real variables.

The second set consists of 90 \smtlarat{} instances encoding temporal planning problems from the industrial-inspired ``Painter'' domain~\cite{valentiniTemporalPlanningIntermediate2020}. For each problem, we generated the SMT encoding for action applications up to a fixed horizon $h\in\set{3,4,5}$ using the \textsc{Tempest} tool~\footnote{\url{https://github.com/fbk-pso/tempest}}. The encoding is described in~\cite{panjkovicExpressiveOptimalTemporal2023}; since~\cite{panjkovicExpressiveOptimalTemporal2023} focuses on optimal temporal planning, in our encoding the optimization component and the abstract step are discarded.
We also removed the goal constraints to increase the number of solutions and make the problems more challenging for \T-lemma enumeration.

\begin{figure}[t]
  \centering
  \begin{subfigure}{0.24\textwidth}
    \centering
    \includegraphics[width=\linewidth, alt={Scatter plot comparing the runtime of following strategies: baseline vs divide-and-conquer.}]
    {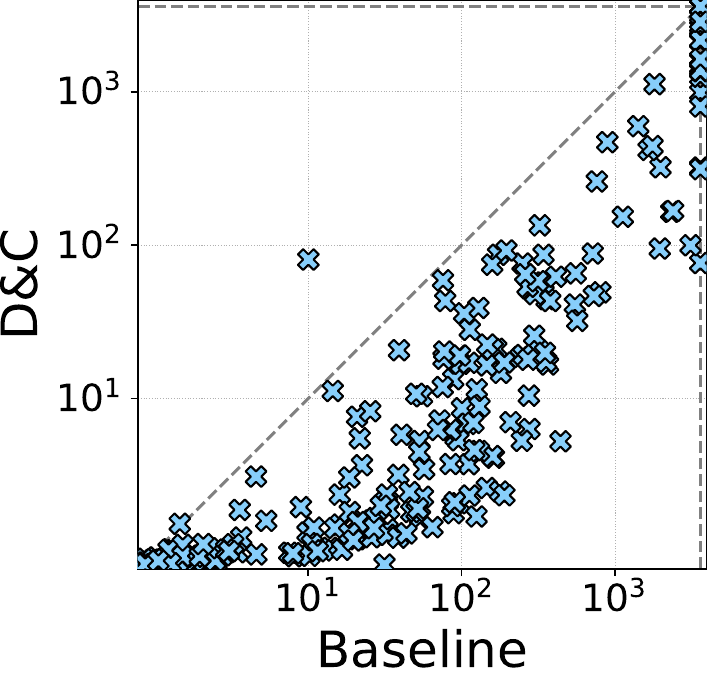}
  \end{subfigure}
  \hfill
  \begin{subfigure}{0.24\textwidth}
    \centering
    \includegraphics[width=\linewidth, alt={Scatter plot comparing the runtime of the following strategies: divide-and-conquer vs divide-and-conquer with projection.}]
    {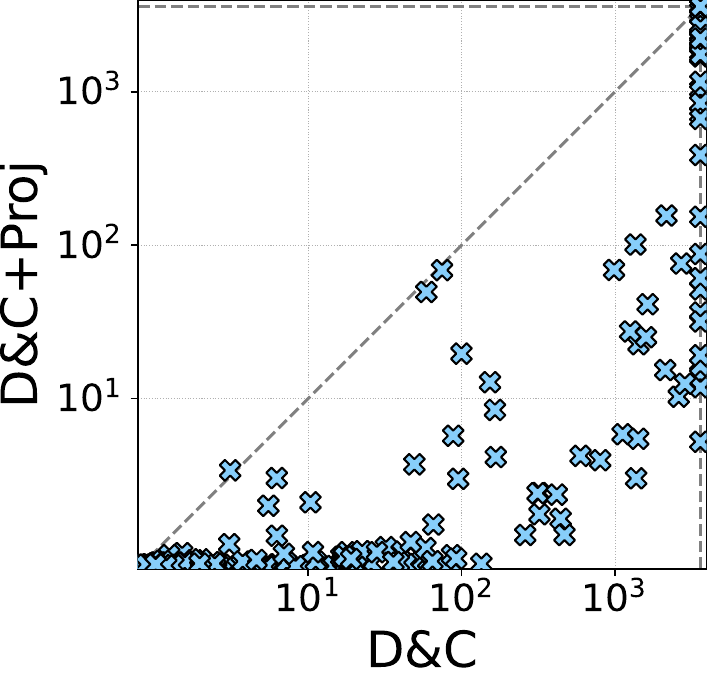}
  \end{subfigure}
  \hfill
  \begin{subfigure}{0.24\textwidth}
    \centering
    \includegraphics[width=\linewidth, alt={Scatter plot comparing the runtime of the following strategies: divide-and-conquer with projection vs divide-and-conquer with projection and partitioning.}]
    {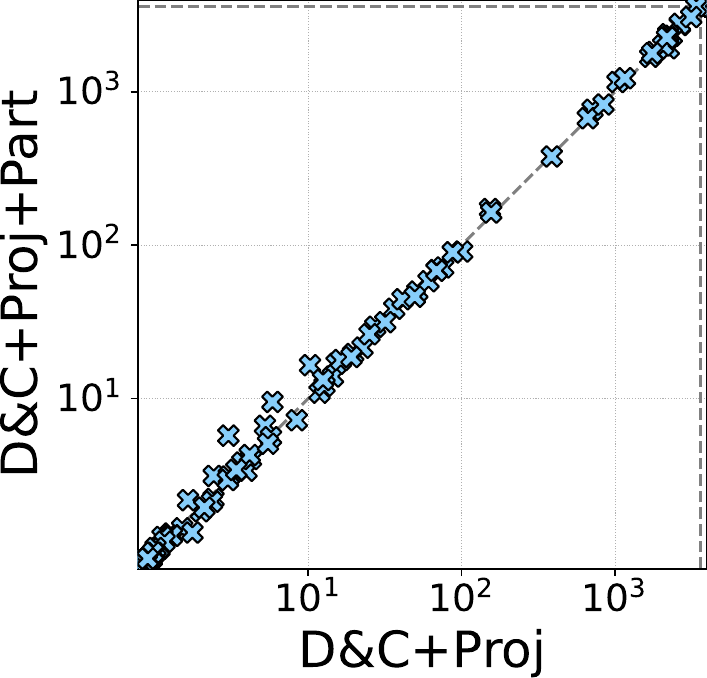}
  \end{subfigure}
  \hfill
  \begin{subfigure}{0.24\textwidth}
    \centering
    \includegraphics[width=\linewidth, alt={Scatter plot comparing the runtime of the following strategies: baseline vs divide-and-conquer with projection and partitioning.}]
    {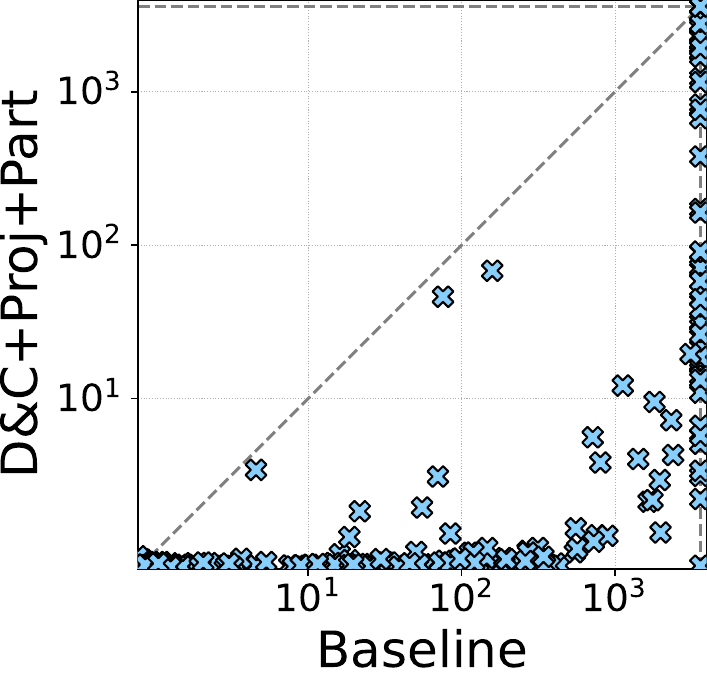}
  \end{subfigure}

  \begin{subfigure}[c]{0.48\textwidth}
    \centering
    \includegraphics[width=\linewidth, alt={Cactus plot comparing the runtime of all the strategies.}]
    {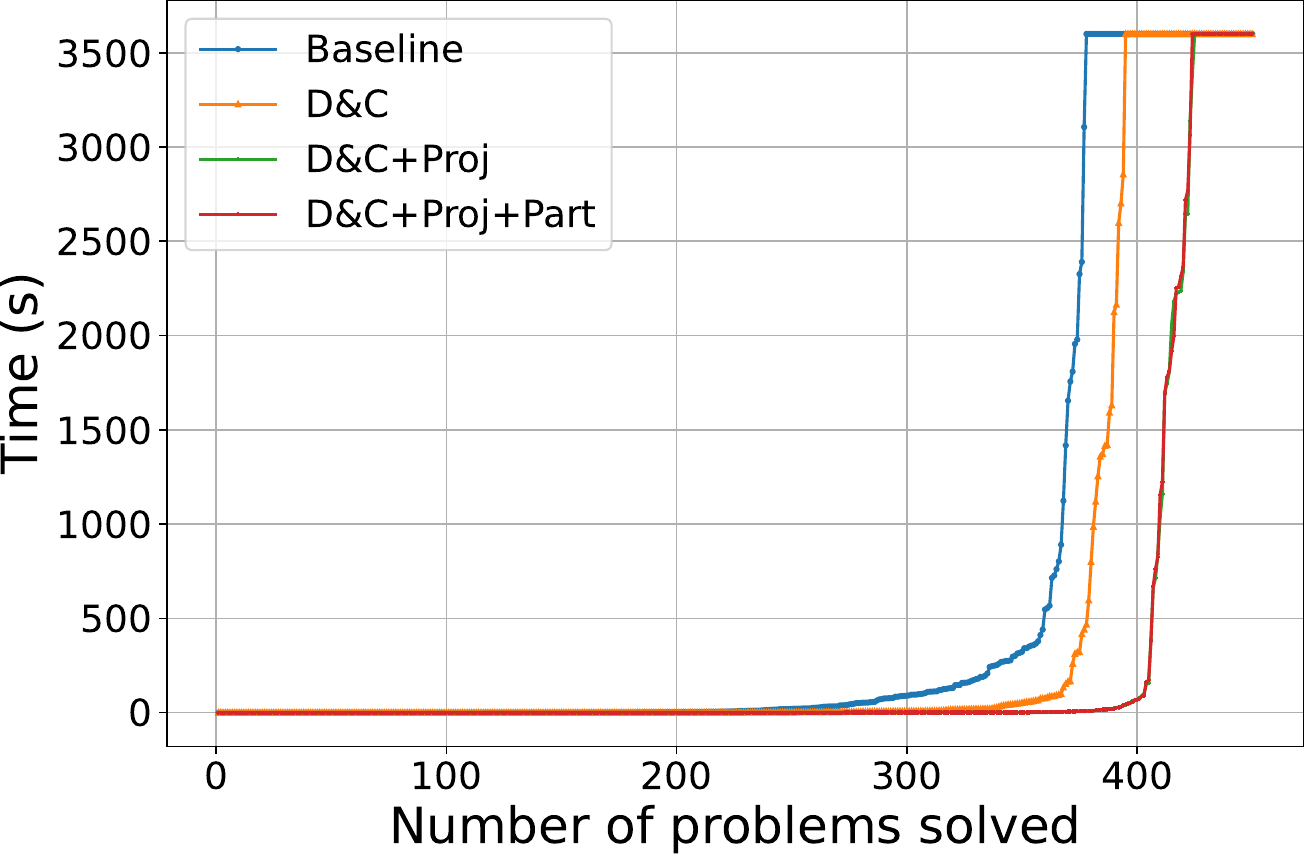}
  \end{subfigure}
  \hfill
  \begin{subfigure}[c]{0.48\textwidth}
    \small
    \centering
    \begin{tabular}{lrr}
      \textbf{Strategy} & \textbf{Timeouts} & \textbf{Total} \\
      \hline
      \expbaseline{}    & 73                & 450            \\
      \expdc{}          & 56                & 450            \\
      \expdcproj{}      & 26                & 450            \\
      \expdcprojpart{}  & 27                & 450            \\
    \end{tabular}
  \end{subfigure}
  \caption{\T-lemma enumeration time with different methods on problems from~\cite{micheluttiCanonicalDecisionDiagrams2024}.
  (Notice that in the cactus plot the green and red lines %
  are almost indistinguishable.)
  }%
  \label{fig:tlemmas-gentime-all-synth}
\end{figure}

\paragraph{Results on enumeration time.}
\Cref{fig:tlemmas-gentime-all-synth,fig:tlemmas-gentime-all-planning} present the results for enumeration time on the synthetic and temporal planning benchmarks, respectively. 
Each figure is organized as follows: the top row shows scatter plots comparing the strategies pairwise. (Notice the logarithmic scale of the axes!) We focus on incremental improvements, i.e., \expbaseline{} vs.\ \expdc, \expdc{} vs.\ \expdcproj, and \expdcproj{} vs.\ \expdcprojpart. The rightmost plot shows the overall speed-up (\expbaseline{} vs.\ \expdcprojpart). The bottom row contains a cactus plot comparing all strategies together, and a table summarizing the number of instances and timeouts.

In the problems from~\cite{micheluttiCanonicalDecisionDiagrams2024}
(\Cref{fig:tlemmas-gentime-all-synth}), we notice that the \expdc{}
significantly improves wrt.\ the baseline by up to 2 orders of
magnitude, solving 17 more problems within the timeout. The improvement is even more pronounced when we add projection, which allows for solving 30 more problems %
and drastically reduces the enumeration time %
for problems already solved by \expdc{}. 
Notably, the introduction of partitioning does not yield any improvement. This is because only 92 of 450 problems have at least two partitions. The hardest of these problems was solved by \expdcproj{} in less than 1 second; thus, partitioning has no room for improvement.
Overall, the scatter plot in \Cref{fig:tlemmas-gentime-all-synth} (top
right) is remarkable, showing a dramatic reduction in the time required by our best configuration wrt.\ the baseline.

\begin{figure}[t]
  \centering
  \begin{subfigure}{0.24\textwidth}
    \centering
    \includegraphics[width=\linewidth, alt={Scatter plot comparing the runtime of the following strategies: baseline vs divide-and-conquer.}]
    {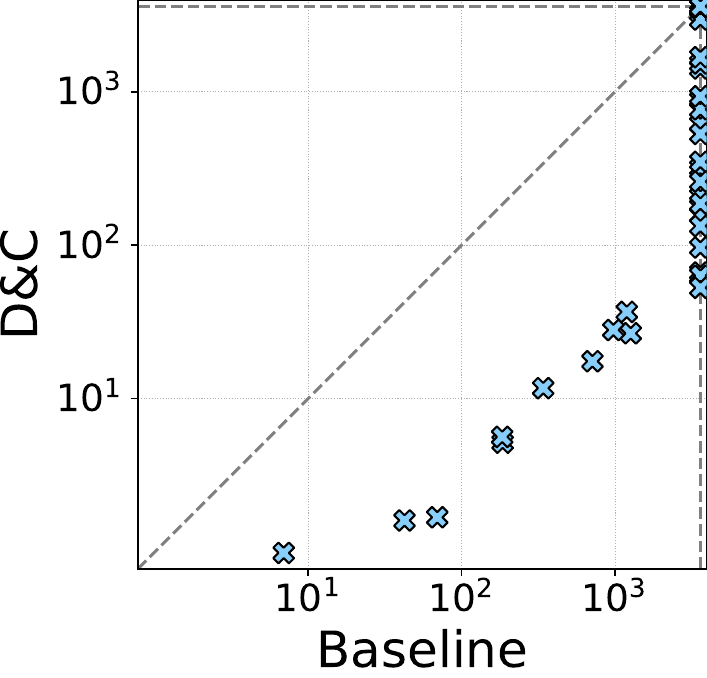}
  \end{subfigure}
  \hfill
  \begin{subfigure}{0.24\textwidth}
    \centering
    \includegraphics[width=\linewidth, alt={Scatter plot comparing the runtime of the following strategies: divide-and-conquer vs divide-and-conquer with projection.}]
    {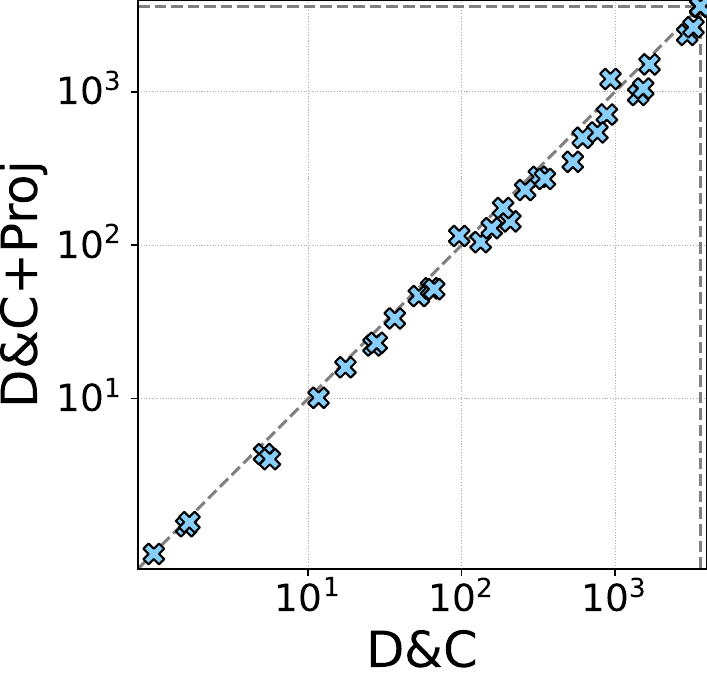}
  \end{subfigure}
  \hfill
  \begin{subfigure}{0.24\textwidth}
    \centering
    \includegraphics[width=\linewidth, alt={Scatter plot comparing the runtime of the following strategies: divide-and-conquer with projection vs divide-and-conquer with projection and partitioning.}]
    {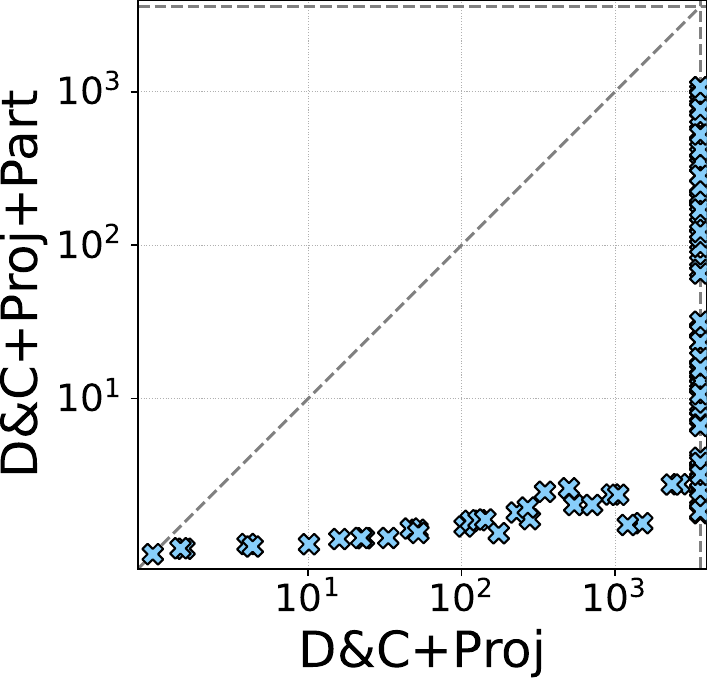}
  \end{subfigure}
  \hfill
  \begin{subfigure}{0.24\textwidth}
    \centering
    \includegraphics[width=\linewidth, alt={Scatter plot comparing the runtime of the following strategies: baseline vs divide-and-conquer with projection and partitioning.}]
    {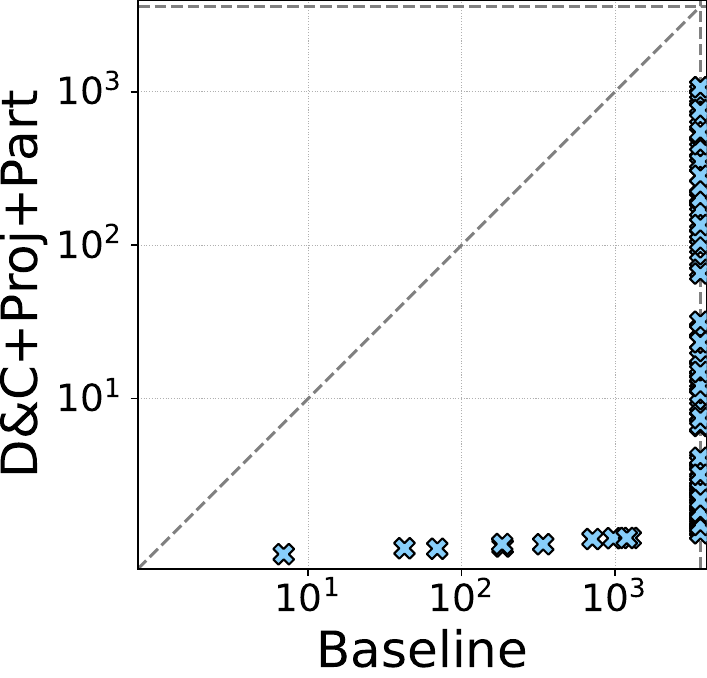}
  \end{subfigure}

  \begin{subfigure}[c]{0.48\textwidth}
    \centering
    \includegraphics[width=\linewidth, alt={Cactus plot comparing the runtime of all the strategies.}]
    {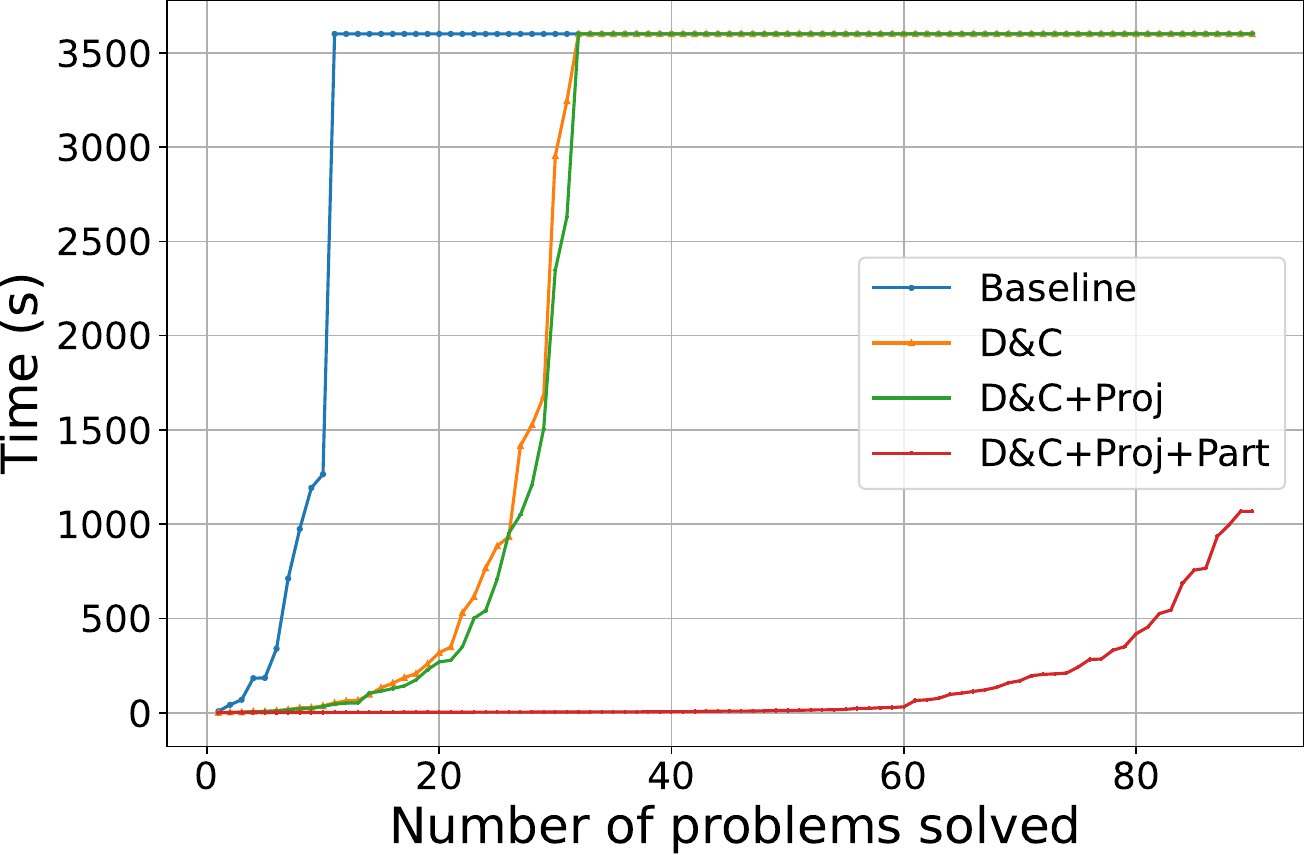}
  \end{subfigure}
  \hfill
  \begin{subfigure}[c]{0.48\textwidth}
    \small
    \centering
    \begin{tabular}{lrr}
      \textbf{Strategy} & \textbf{Timeouts} & \textbf{Total} \\
      \hline
      \expbaseline        & 80                & 90            \\
      \expdc              & 59                & 90            \\
      \expdcproj          & 59                & 90            \\
      \expdcprojpart      & 0                 & 90            \\
    \end{tabular}
  \end{subfigure}
  \caption{\T-lemma enumeration time with different methods on planning problems.}%
  \label{fig:tlemmas-gentime-all-planning}
\end{figure}

On temporal planning problems (\Cref{fig:tlemmas-gentime-all-planning}), we can see a similar trend, with some differences. While \expdc{} is much faster than the baseline, solving 21 more problems within the timeout, the addition of projection does not yield any improvement. This is likely due to the fact that in planning problems, the assignment on \T-atoms is a deterministic consequence of the assignment on \B-atoms, so that projecting does not reduce the number of \T-satisfiable truth assignments explored by the solver.
On the other hand, partitioning is extremely effective here, dramatically reducing the enumeration time. 
In these %
problems, we have five \T-atoms partitions: atoms encoding time constraints, atoms with constraints on the values of the fluents, and three partitions with constraints on the parameters of three different actions.  
By reasoning separately on these partitions, we achieve a speedup of several orders of magnitude, solving problems up to horizon 5 that were before intractable.

\paragraph{Results on the number of enumerated \T-lemmas.}
\begin{figure}[t]
  \centering
  \begin{subfigure}{0.24\textwidth}
    \centering
    \includegraphics[width=\linewidth, alt={Scatter plot comparing the number of \T-lemmas generated by the following strategies: baseline vs divide-and-conquer.}]
    {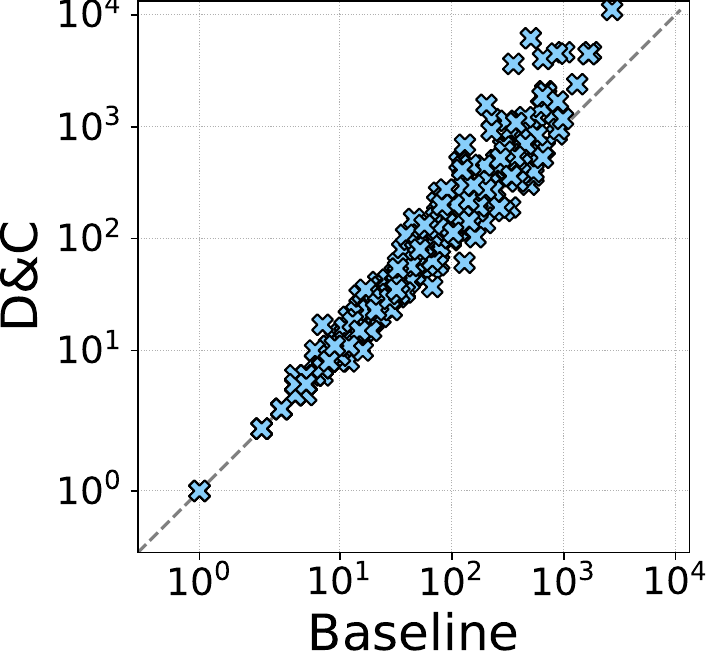}
  \end{subfigure}
  \hfill
  \begin{subfigure}{0.24\textwidth}
    \centering
    \includegraphics[width=\linewidth, alt={Scatter plot comparing the number of \T-lemmas generated by the following strategies: divide-and-conquer vs divide-and-conquer with projection.}]
    {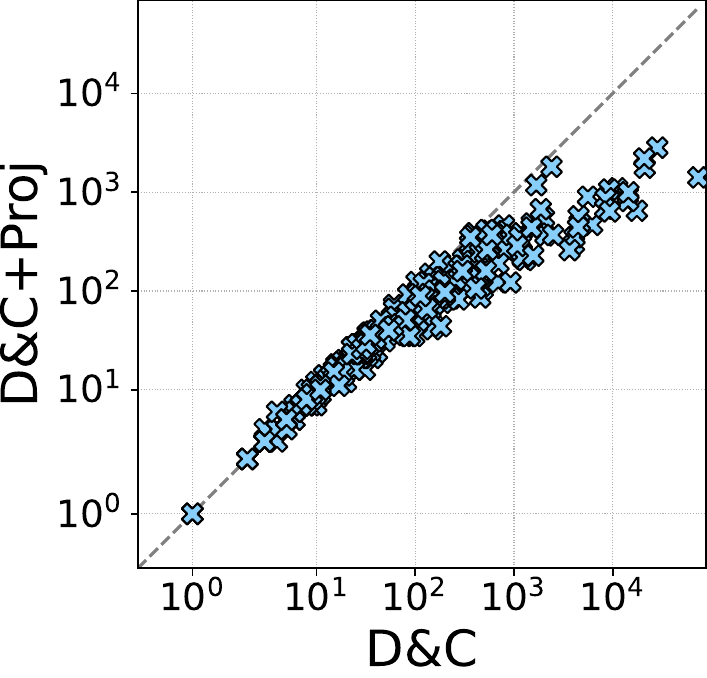}
  \end{subfigure}
  \hfill
  \begin{subfigure}{0.24\textwidth}
    \centering
    \includegraphics[width=\linewidth, alt={Scatter plot comparing the number of \T-lemmas generated by the following strategies: divide-and-conquer with projection vs divide-and-conquer with projection and partitioning.}]
    {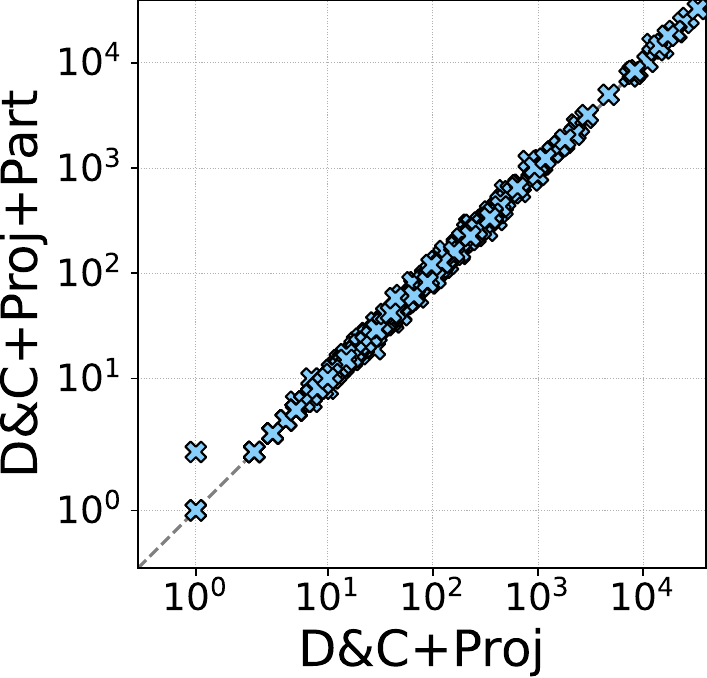}
  \end{subfigure}
  \hfill
  \begin{subfigure}{0.24\textwidth}
    \centering
    \includegraphics[width=\linewidth, alt={Scatter plot comparing the number of \T-lemmas generated by the following strategies: baseline vs divide-and-conquer with partitioning.}]
    {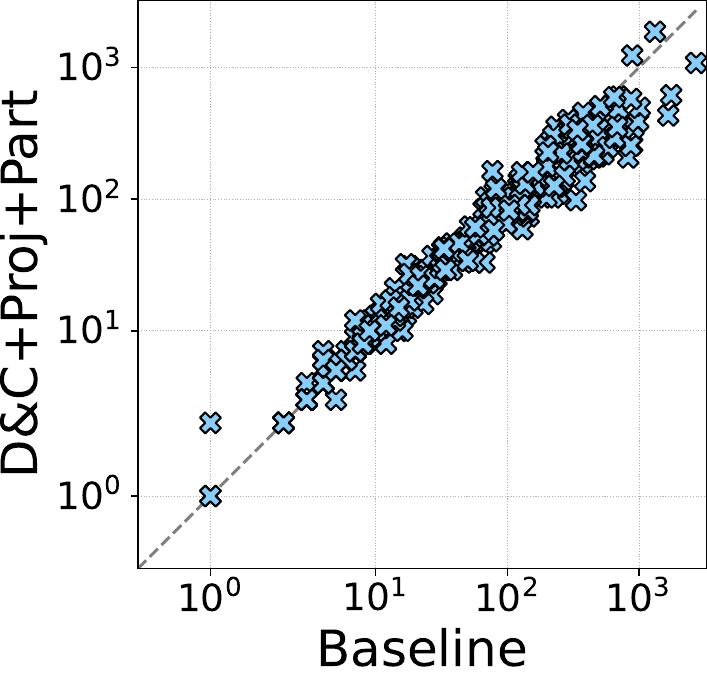}
  \end{subfigure}
  \caption{Number of \T-lemmas found by different strategies on problems from~\cite{micheluttiCanonicalDecisionDiagrams2024}.}%
  \label{fig:tlemmas-num-all-synth}
  \centering
  \begin{subfigure}{0.24\textwidth}
    \centering
    \includegraphics[width=\linewidth, alt={Scatter plot comparing the number of \T-lemmas generated by the following strategies: baseline vs divide-and-conquer.}]
    {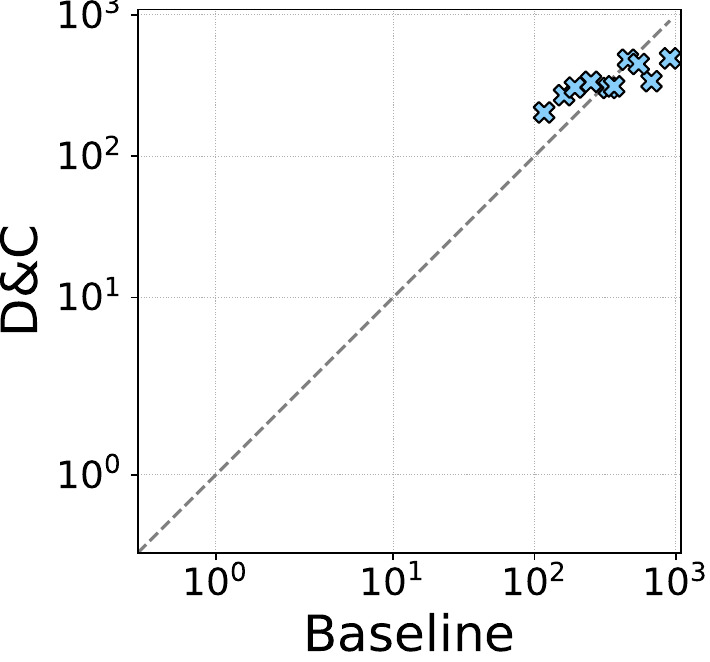}
  \end{subfigure}
  \hfill
  \begin{subfigure}{0.24\textwidth}
    \centering
    \includegraphics[width=\linewidth, alt={Scatter plot comparing the number of \T-lemmas generated by the following strategies: divide-and-conquer vs divide-and-conquer with projection.}]
    {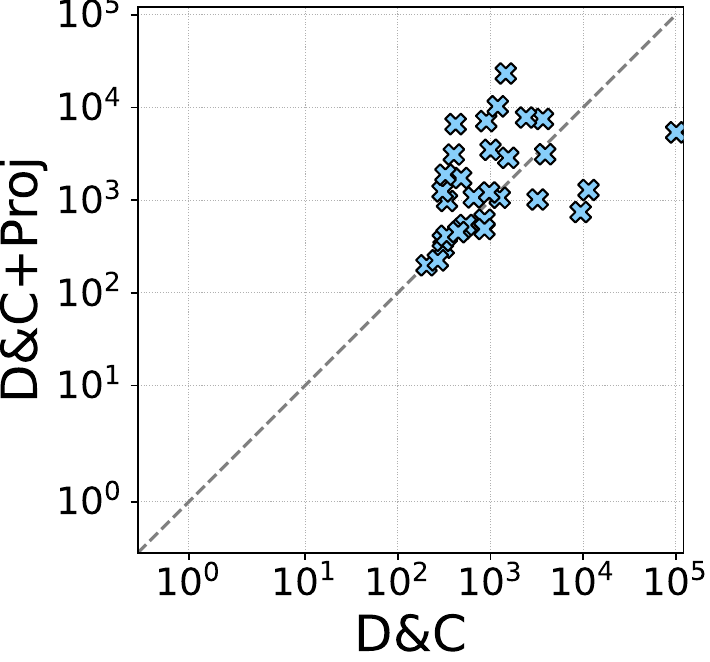}
  \end{subfigure}
  \hfill
  \begin{subfigure}{0.24\textwidth}
    \centering
    \includegraphics[width=\linewidth, alt={Scatter plot comparing the number of \T-lemmas generated by the following strategies: divide-and-conquer with projection vs divide-and-conquer with projection and partitioning.}]
    {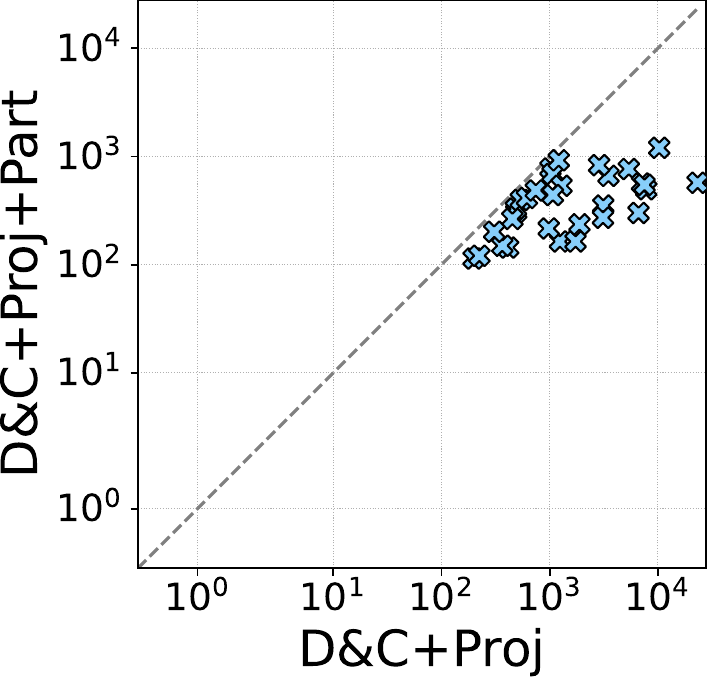}
  \end{subfigure}
  \hfill
  \begin{subfigure}{0.24\textwidth}
    \centering
    \includegraphics[width=\linewidth, alt={Scatter plot comparing the number of \T-lemmas generated by the following strategies: baseline vs divide-and-conquer with partitioning.}]
    {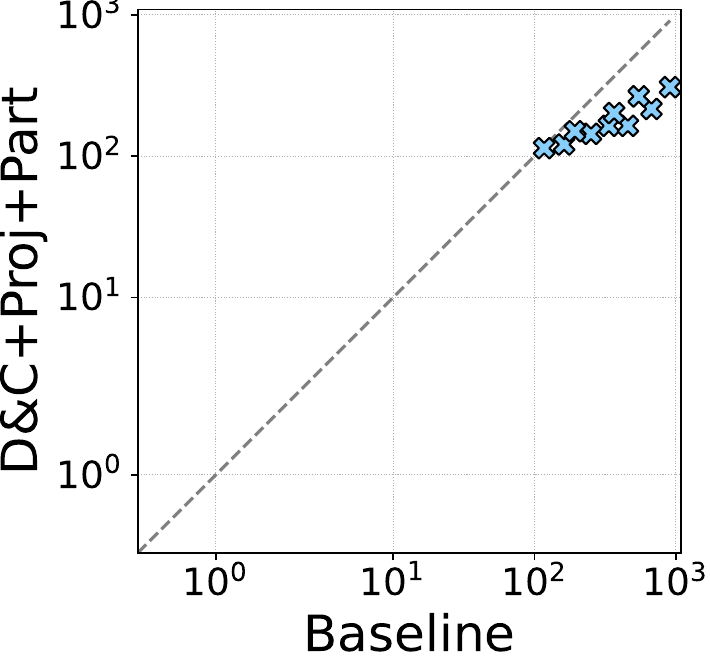}
  \end{subfigure}
  \caption{Number of \T-lemmas found by different strategies on planning problems.}%
  \label{fig:tlemmas-num-all-planning}
\end{figure}

We now analyze the number of \T-lemmas enumerated by the different strategies. \Cref{fig:tlemmas-num-all-synth,fig:tlemmas-num-all-planning} show scatter plots for the synthetic and temporal planning benchmarks, respectively.
As for the time plots, also here we present pairwise comparisons of incremental improvements. We remark that we only show points for problems solved by both the compared strategies, so that plots in \Cref{fig:tlemmas-num-all-planning} only have a limited number of points.

Overall, %
there is no big difference in the number of \T-lemmas enumerated by different strategies. We can, however, observe some trends.

First, \expdc{} tends to enumerate more \T-lemmas than \expbaseline{}. This is expected as we said in~\sref{sec:divide-and-conquer}, since the parallel %
enumeration %
can %
discover %
redundant \T-lemmas.

Second, both projection and partitioning help in reducing the number of enumerated \T-lemmas.
This fact can have multiple explanations.

One is that early pruning to the \T-solver~\cite{barrettSatisfiabilityModuloTheories2021} may cause the SMT solver to enumerate \T-lemmas that rule out also \T-unsatisfiable assignments that do not propositionally satisfy the formula. By narrowing the search space, both projection and partitioning help mitigate this undesired behavior.

Additionally, for the divide-and-conquer strategy without projection, two parallel processes can find two \T-inconsistent assignments sharing the same \T-literals and different \B-literals, inevitably generating two redundant \T-lemmas. With projection, instead, the partial assignments fed to different processes differ only on \T-literals, thereby mitigating this issue.

Lastly, %
partitioning can help enumerate smaller \T-lemmas (see next), which have stronger pruning power, thus reducing the number of enumerated \T-lemmas.

\paragraph{Results on the size of \T-lemmas.}%
\label{app:experiments:median-literals}
As a further insight, %
we analyze the median size (i.e., the number of literals)
of the \T-lemmas enumerated by the different strategies. The results are shown in \Cref{fig:tlemmas-mediansize-all-synth,fig:tlemmas-mediansize-all-planning} for synthetic and planning problems, respectively.
Since we only show points for problems solved by both the compared
strategies, so that plots in \Cref{fig:tlemmas-mediansize-all-planning} only have a limited number of points.%

We observe that different techniques do not consistently improve or worsen the size of \T-lemmas, with a notable exception for partitioning, which, on planning problems, allows for finding much smaller \T-lemmas.
We conjecture that this may be due to the fact that, with projected
enumeration, the literal selection heuristics of the SMT solver may
prefer selecting projection atoms first.
Hence, early-pruning
calls~\cite{barrettSatisfiabilityModuloTheories2021} to the \T-solver{} are more likely
to cut \T-inconsistent branches involving only projection atoms.
This hinders the generation of bigger \T-lemmas with atoms from
different partitions, which are
intrinsically redundant by~\Cref{prop:partitioning}. 
\begin{figure}[t]
    \centering
    \begin{subfigure}{0.24\textwidth}
        \centering
        \includegraphics[width=\linewidth, alt={Scatter plot comparing the median size of \T-lemmas generated by the following strategies: baseline vs divide-and-conquer.}]
        {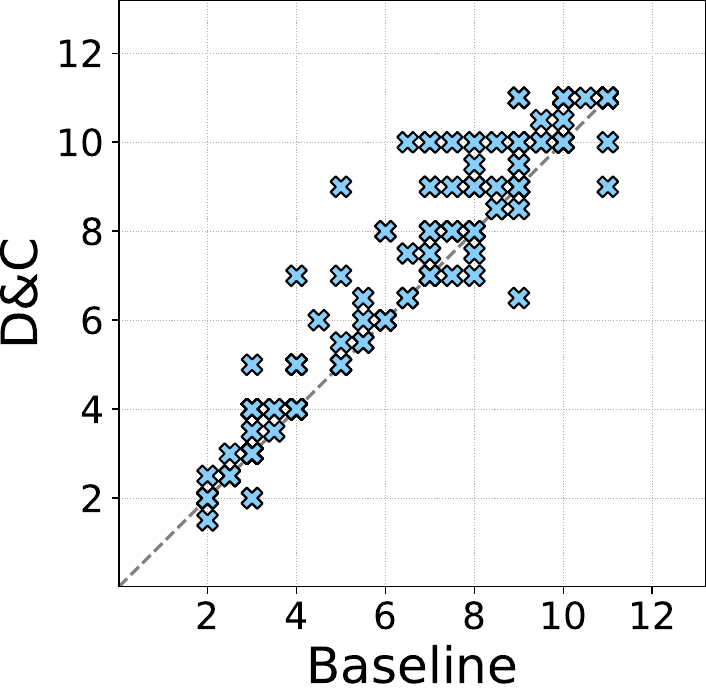}
    \end{subfigure}
    \hfill
    \begin{subfigure}{0.24\textwidth}
        \centering
        \includegraphics[width=\linewidth, alt={Scatter plot comparing the median size of \T-lemmas generated by the following strategies: divide-and-conquer vs divide-and-conquer with projection.}]
        {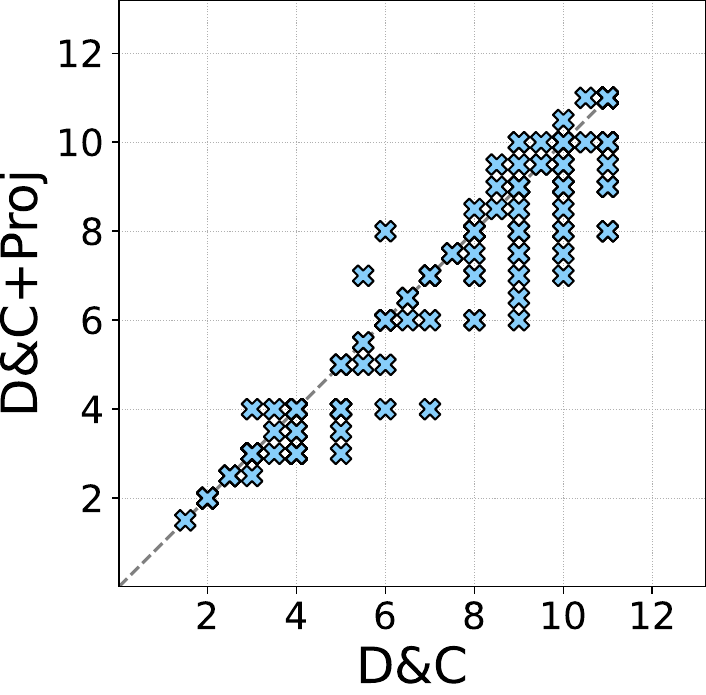}
    \end{subfigure}
    \hfill
    \begin{subfigure}{0.24\textwidth}
        \centering
        \includegraphics[width=\linewidth, alt={Scatter plot comparing the median size of \T-lemmas generated by the following strategies: divide-and-conquer with projection vs divide-and-conquer with projection and partitioning.}]
        {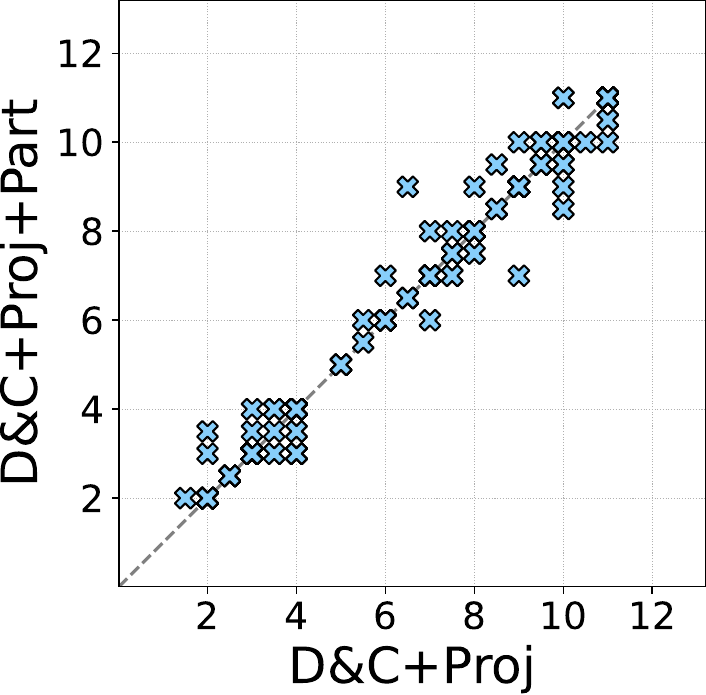}
    \end{subfigure}
    \hfill
    \begin{subfigure}{0.24\textwidth}
        \centering
        \includegraphics[width=\linewidth, alt={Scatter plot comparing the median size of \T-lemmas generated by the following strategies: baseline vs divide-and-conquer with partitioning.}]
        {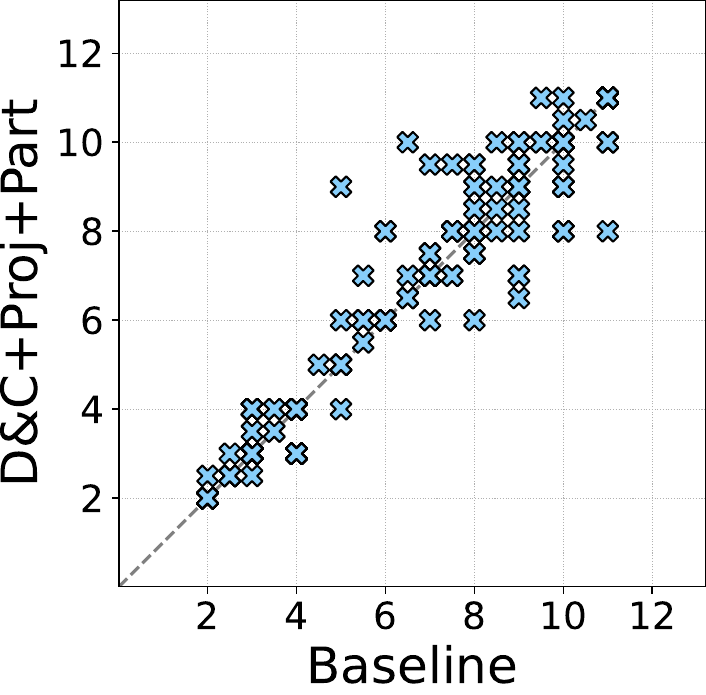}
    \end{subfigure}
    \caption{Median \T-lemma size for different strategies on problems from~\cite{micheluttiCanonicalDecisionDiagrams2024}.}
    \label{fig:tlemmas-mediansize-all-synth}

    \centering
    \begin{subfigure}{0.24\textwidth}
        \centering
        \includegraphics[width=\linewidth, alt={Scatter plot comparing the median size of \T-lemmas generated by the following strategies: baseline vs divide-and-conquer.}]
        {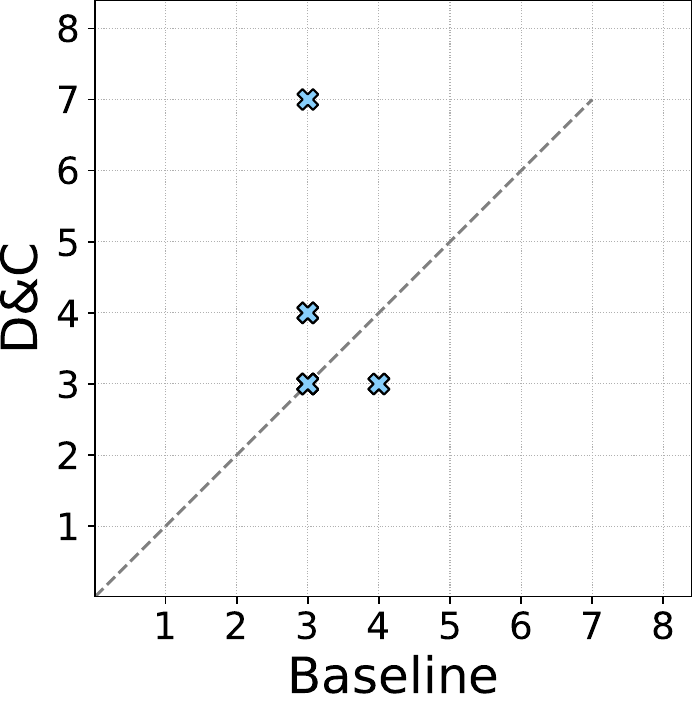}
    \end{subfigure}
    \hfill
    \begin{subfigure}{0.24\textwidth}
        \centering
        \includegraphics[width=\linewidth, alt={Scatter plot comparing the median size of \T-lemmas generated by the following strategies: divide-and-conquer vs divide-and-conquer with projection.}]
        {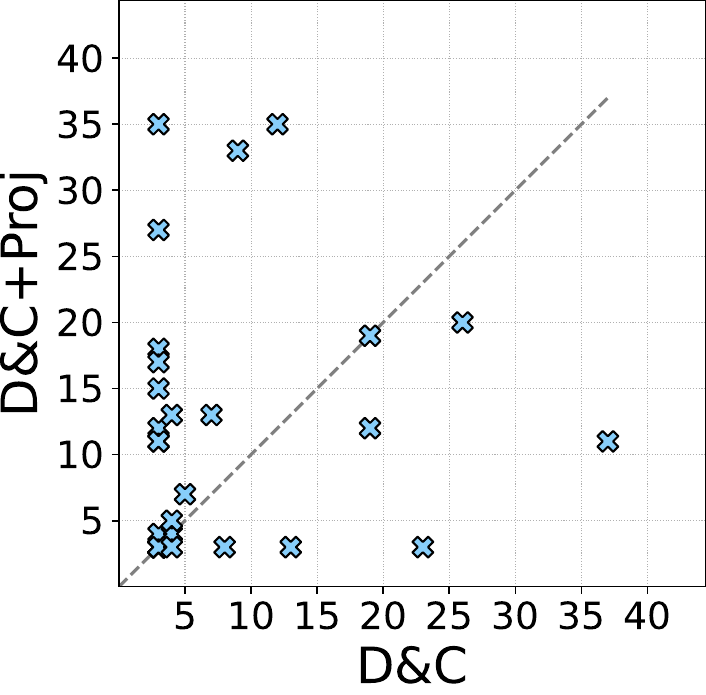}
    \end{subfigure}
    \hfill
    \begin{subfigure}{0.24\textwidth}
        \centering
        \includegraphics[width=\linewidth, alt={Scatter plot comparing the median size of \T-lemmas generated by the following strategies: divide-and-conquer with projection vs divide-and-conquer with projection and partitioning.}]
        {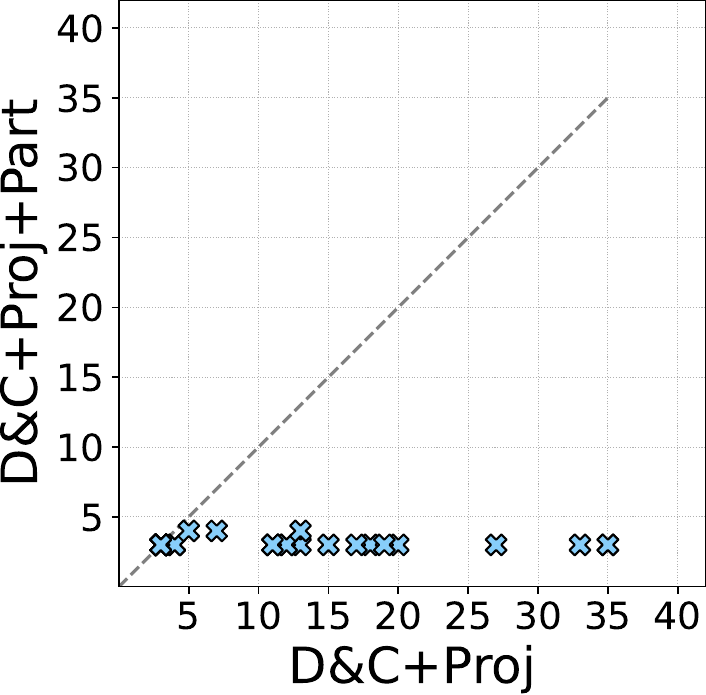}
    \end{subfigure}
    \hfill
    \begin{subfigure}{0.24\textwidth}
        \centering
        \includegraphics[width=\linewidth, alt={Scatter plot comparing the median size of \T-lemmas generated by the following strategies: baseline vs divide-and-conquer with partitioning.}]
        {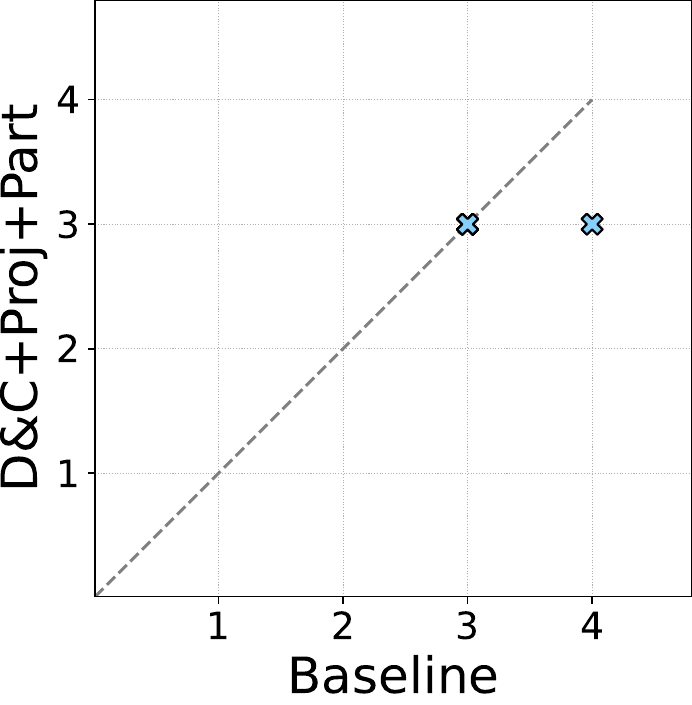}
    \end{subfigure}
    \caption{Median \T-lemma size for different strategies on planning problems.}
    \label{fig:tlemmas-mediansize-all-planning}
\end{figure}

\label{sec:experiments}

\section{Conclusions and Future Work}%
\label{sec:conclusions}
In this paper, we presented theory-agnostic techniques for enumerating a set of theory lemmas ruling out all theory-inconsistent truth assignments for a given SMT formula.
Starting from a baseline method based on total AllSMT enumeration proposed in~\cite{micheluttiCanonicalDecisionDiagrams2024}, we proposed three orthogonal, composable, techniques to enhance efficiency.

We experimentally evaluated our techniques on SMT problems
from~\cite{micheluttiCanonicalDecisionDiagrams2024}, and on a new set
of SMT formulas encoding temporal planning
problems~\cite{valentiniTemporalPlanningIntermediate2020}. Our
techniques drastically improve the efficiency of lemma enumeration
wrt.\ the baseline, allowing us to scale to much more complex instances.

This work is planned to play a crucial role in %
our ongoing research in {\em knowledge compilation modulo theories}, which we
introduced in~\cite{micheluttiCanonicalDecisionDiagrams2024} for
decision diagrams and %
in~\cite{masina-sat26-ddnnf} for general d-DNNFs.
Since in~\cite{micheluttiCanonicalDecisionDiagrams2024} the
enumeration of theory lemmas was the
major bottleneck in the compilation process, our novel lemma-enumeration techniques pave
the way for scaling to complex, real-world instances.

\FloatBarrier

\newpage

\begin{credits}
	\subsubsection{\ackname}
	We thank Alberto Griggio for his assistance with the \mathsat{} usage.
	We thank Alessandro Valentini and Andrea Micheli for their assistance with \textsc{Tempest} for generating the temporal planning benchmarks.

\end{credits}

\appendix

\crefalias{section}{appendix}
\crefalias{subsection}{appendix}

\section{Proofs of the Theorems}%
\label{app:proofs}
\subsection{Proof of \Cref{teo:lemmasout-exists}}%
\label{app:proofs:teo:lemmasout-exists}

\begin{proof}

  Since $\BCAT{\vi}\cup\TLEMMASANY$ rules out $\ITTA{\vi}$, and since $\BCAT{\vi}$ rules out only \T-satisfiable assignments on $\allalphaT{}$, 
  then every $\rhoa{}\defas\rhoaT{}\wedge\rhoaB{}\in\ITTA{\vi}$
  s.t.\ $\rhoaT{}$ is \T-unsatisfiable is ruled out by some \T-lemma
  in $\TLEMMASANY$.

  Since $\rhoaT{}\wedge\rhoaB{}$ is \T-unsatisfiable iff $\rhoaT{}$ is \T-unsatisfiable, then every $\rhoa{}\in\ITTA{\vi}$ is ruled out by some \T-lemma in $\TLEMMASANY$.

\end{proof}

\subsection{Proof of \Cref{teo:lemmasout-partitioning}}%
\label{app:proofs:teo:lemmasout-partitioning}
\begin{proof}

  Since $\BCARED{\vi}\cup\TLEMMASANY$ rules out $\ITTA{\vi}$, and since $\BCARED{\vi}$ rules out only \T-satisfiable assignments on $\allalphared{}$, 
  then every $\rhoa{}\defas\rhoared{}\wedge\rhoablu{}\in\ITTA{\vi}$
  s.t.\ $\rhoared{}$ is \T-unsatisfiable is ruled out by some \T-lemma
  in $\TLEMMASANY$.

  Consequently, due to the fact that all \T-lemmas  in $\TLEMMASANY$ have been added
  to \viprime{}, every
  $\rhoprimea{}\defas\rhoprimeared{}\wedge\rhoprimeablu{}\in\ITTA{\viprime{}}$
  is such that $\rhoprimeared{}$ is \T-satisfiable.

  Thus, since $\allalphared$ and $\allalphablu$ do not share variables or uninterpreted symbols, by \Cref{prop:partitioning}, every $\rhoprimea{}\defas\rhoprimeared{}\wedge\rhoprimeablu{}\in\ITTA{\viprime{}}$ is such that $\rhoprimeablu{}$ is \T-unsatisfiable. Since $\BCABLU{\vi}$ only rules out \T-satisfiable assignments on $\allalphablu{}$, then $\TLEMMASANYPRIME$ rules out every $\rhoprimea{}\in\ITTA{\viprime{}}$,
  that is, every $\rhoprimea{}\in\ITTA{\via{}}$ which was not already
  ruled out by the \Tlemmas in  $\TLEMMASANY$.
  
  Hence, $\TLEMMASANY\cup\TLEMMASANYPRIME$ rules out $\ITTA{\vi}$.
\end{proof}


\bibliographystyle{splncs04}
\bibliography{bibliography}
\end{document}